\documentclass[12pt,a4paper,final]{article}
\usepackage[margin=23mm]{geometry}
\usepackage[utf8]{inputenc}
\usepackage{hyperref}
\usepackage{amsmath,amsfonts,amssymb,bef_alex,%
            datetime,todonotes,bm,relsize,tikz,mathrsfs}
\renewcommand{\ti}{{\times}}

\numberwithin{equation}{section}
\numberwithin{figure}{section}

\usepackage[notref,notcite,color]{showkeys}

\DeclareMathOperator{\trace}{\mafo{tr}}

\newcommand{\Bh}{{\mathsf B\mathsf h}}

\newcommand{\MM}{\mathfrak M(\Omega)}
\newcommand{\MMM}[1]{\mathfrak M(\Omega_{#1})}
\newcommand{\PP}{\mathfrak P(\Omega)}
\newcommand{\PPP}[1]{\mathfrak P(\Omega_{#1})}
\newcommand{\oti}{{\otimes}}

\newcommand{\He}{\mathsf{He}}

\newcommand{\FR}{\mathsf{FR}}

\newcommand{\mfP}{\mathfrak P}

\newcommand{\SEarrow}{\text{\smaller${}\searrow{}$}}

\begin{document}

\title
{Some notes on the Hellinger distance and \\various Fisher-Rao distances%
}

\author
{Alexander Mielke\thanks{Weierstraß-Institut f\"ur Angewandte
   Analysis und Stochastik, Anton-Wilhelm-Amo-Str.\,39, 10117 Berlin
   }
}

\date{2 October 2025}
 
\maketitle


\begin{abstract}  These expository notes introduce the Hellinger distance on the
  set of all measures and the induced Fisher-Rao distances for subsets of measures,
  such as probability measures or Gaussian measures. The historical background
  is highlighted and the relations and the distinct features of the two 
  distances are discussed. Moreover, we provide a dynamic characterization of 
  absolutely continuous curves in the Hellinger spaces in terms of the growth
  equation, which replaces the continuity equation in the theory of optimal
  transport.\medskip

   \smallskip

\end{abstract}



\section{Introduction}
\label{se:intro}

The initial motivation for writing this mainly expository notes was the
question what are the origin of the names ``Hellinger distance'' and
``Fisher-Rao distance''. Hence, we will introduce these two concepts, which are
indeed closely related, in simple terms and explain the historical developments
that go back to Hellinger \cite{Hell07OQFU,Hell09NBTQ} and Kakutani
\cite{Kaku48EIPM} for the Hellinger distance and to Fisher \cite{Fish21MFTS}
and Rao \cite{Rao45IAAE} for the Fisher-Rao distance.  Unfortunately, these
names are sometimes mixed up and these notes provide a guideline for
distinguishing the two objects such that future mathematical discussion can
made more precise by avoiding unnecessary confusion.

A second goal arises from the recent interest in gradient flows in the
Hellinger space (e.g.\ \cite{CCHH24?FRGF, MieZhu25?ECHK}) and in the
combination of the Wasserstein distance and the Hellinger distance in the
transport growth distance called Hellinger-Kantorovich distance in
\cite{LiMiSa16OTCR, LiMiSa18OETP} and and Wasserstein-Fisher-Rao distance in
\cite{CPSV18IDOT, CPSV18UOTG}. While there is a large body of work in
characterizing absolutely continuous curves in the Wasserstein space via the
dynamic theory of Benamou-Brenier \cite{BenBre00CFMS} and Otto
\cite{Otto01GDEE} using the continuity equation, there is no counterpart
available for the absolutely continuous curves in the Hellinger space. In
principle, the corresponding theory can be extracted from the analysis of the
Hellinger-Kantorovich theory in \cite{LiMiSa18OETP}, but this would lead to a
huge and inscrutable overhead. In Section \ref{su:AbsCurves} we provide a
short and mathematically complete characterization, which shows the relations
between the metric derivative and the growth equation, which replaces the
continuity equation.

The Hellinger distance  $\He (\mu_0,\mu_1)$ between arbitrary measures $\mu_0,
\mu_1 \in \MM$ on a measure space $\Omega$ are is defined by 
\begin{equation}
  \label{eq:I.def.Hell}
  \He (\mu_0,\mu_1)^2 = \sigma^2 \int_\Omega \Big( \big(\frac{\rmd\mu_1}{\rmd
    \lambda}\big)^{1/2}  -\big(\frac{\rmd\mu_0}{\rmd
    \lambda}\big)^{1/2}  \Big)^2 \dd \lambda = \sigma^2 \Big( \mu_0(\Omega) +
  \mu_1(\Omega) - 2 \sqrt{\mu_0\mu_1}\,(\Omega) \Big), 
\end{equation}
where $\lambda \in \MM$ is any measure such that $\mu_0\ll \lambda$ and
$\mu_1 \ll \lambda$, and
$\frac{\rmd\mu_j}{\rmd\lambda} \in \rmL^1(\Omega,\lambda)$ denotes the
Radon-Nikodym derivative. Here we have introduced a scaling factor $\sigma>0$
into the definition because in various places in the literature different
factors are chosen. We keep the factor throughout to facility the comparison
with the literature, but it is also helpful to understand the structure better.  

The definition of $\He$ goes back to Kakutani \cite{Kaku48EIPM} and was chosen
to honor the contribution of Hellinger in \cite{Hell07OQFU,Hell09NBTQ} which
showed how to define the Hellinger integral 
\begin{equation}
  \label{eq:RiemHell}
  \sqrt{\mu_0\mu_1}(B) = \int_B  \big(\frac{\rmd\mu_0}{\rmd\lambda}\: 
    \frac{\rmd\mu_1}{\rmd\lambda}\big)^{1/2}  \dd \lambda 
    \quad \text{for measurable } B\subset \Omega,
\end{equation}
much before the introduction of Radon-Nikodym derivatives. Thus, the geometric
mean $\sqrt{\mu_0\mu_1}$ is again a well-defined measure in $\MM$. Moreover,
\cite{Kaku48EIPM} showed that for every $\lambda\in \MM$ the Hilbert space
$\rmL^2(\Omega,\lambda)$ can be isometrically embedded into $(\MM,\He)$
via the mapping $\rmL^2(\Omega,\lambda) \ni g \mapsto \sqrt{g} \,\lambda \in
\MM$. This embedding immediately shows that 
the Hellinger distance is even a geodesic distance in the sense that for every
pair $(\mu_0,\mu_1)$ there exists a unique constant-speed geodesic curve given
by 
\begin{equation}
  \label{eq:I.HellGeod}
 \begin{aligned}
 \gamma^\He_{\mu_0\to\mu_1}(s) &= (1{-}s)^2\mu_0 + s^2 \mu_1 +2(s{-}s^2)
\sqrt{\mu_0\mu_1}\\
&= (1{-}s) \mu_0 + s\mu_1 -(s{-}s^2) \frac1{\sigma^2} \He(\mu_0,\mu_1)^2 \quad
\text{for } s\in [0,1]. 
\end{aligned}
\end{equation}
Moreover, one can define a pseudo-Riemannian structure on $\MM$ given by the
quadratic form of Hellinger type (see \eqref{eq:IntegralHelli} in the historical
remarks in Section \ref{su:History})
\begin{equation}
  \label{eq:I.QuadrForm}
  \bfg_\mu(\nu_1,\nu_2) = \left\{ \ba{cl} \ds \frac{\sigma^2}4 
 \int_\Omega  \frac{\rmd\nu_1}{\rmd\mu}\:\frac{\rmd\nu_2}{\rmd\mu}\:\dd \mu &
 \text{if } \nu_1,\nu_2 \ll \mu, \\ 
\infty & \text{else}. \ea\right. 
\end{equation}
In Theorem \ref{th:AbsContCurve} we provide the mathematically connection of
$\bfg_\mu$ with the metric speed in $(\MM,\He)$. Moreover, we refer to
\cite{AJLS15IGSS,BaBrMi16UFRM} for a proof of the uniqueness of this Riemannian
metric under diffeomorphisms.

The above Riemannian structure was indeed introduced by Fisher in
\cite{Fish21MFTS} when studying finite-dimensional parameterized families of
measures. Considering the family
\[
\calS = \bigset{ f(p;\cdot) \in \rmL^1(\R^n)\cap \mfP(\R^n) }{ p \in D\subset \R^m}
\]
Then, the Fisher information metric is defined via the matrix  $\bbF(p)\in
\R^{m\ti m}_{\geq 0}$  given by 
\begin{equation}
  \label{eq:FishMetric}
\begin{aligned}
  a\cdot \bbF(p)b &:=  \frac{\sigma^2}4 \int_{\R^n} \rmD_p \log\big(f(p,x) 
            \big)[a] \:\rmD_p\log\big(f(p,x)\big) [b] \: f(p,x) \dd x\\
& =   \frac{\sigma^2}4   \int_{\R^n} \frac{\rmD_p f(p,x) [a] \,\rmD_p
         f(p,x) [b]} {f(p,x)}  \dd x \\
& = - \frac{\sigma^2}4  \int_{\R^m} \rmD^2_p \log\big(f(p,x)\big)[a,b] \:f(p,x)\dd x 
\end{aligned}
\end{equation}
The Fisher-Rao distance was introduced in \cite{Rao45IAAE} and is defined as the
distance on $D$ induced by the metric tensor $\bbF$, namely 
\[
\FR_\calS(p_0,p_1): =\inf\Bigset{\int_0^1\!\!\! \big( p'(s)\cdot \bbF(p(s))p')s)
  \big)^{1/2} \dd s   }{ p\in \rmC^1([0,1];D),\ p(0)=p_0,\ p(1)=p_1}.
\]
The strength of the Fisher-Rao distance is that it does not depend on the
particular choice of the pasteurization, but only on the subset $\calS
\subset \MM$. Thus, we will $\FR_\calS\big(f(p_0,\cdot)\rmd x,f(p_1,\cdot)\rmd
x)\big) $ instead of $\FR_\calS(p_0,p_1)$. 

The above construction is not restricted to finite dimensional submanifolds,
but can be generalized to more general subsets $\calS$ of $\MM$. Indeed, 
$\FR_\calS:\calS\ti \calS \to[0,\infty]$ can be understood as the intrinsic
length or distance in $\calS$ induced from $(\MM,\He)$. For this, we consider
continuous paths $\gamma:[0,1]\to\calS$, define their intrinsic length
$L_\He(\gamma)$, see \eqref{eq:IntrLength}, and then define the
Fisher-Rao distance in $\calS$ by 
\[
\FR_\calS(\mu_0,\mu_1) := \inf\Bigset{L_\He(\gamma) }{ \gamma\in
  \rmC^0([0,1];\calS),\ \gamma(0)=\mu_0,\ \gamma(1)=\mu_1) }.
\]
Of course choosing $\calS=\MM$, i.e., the space of all measures, we
have $\FR_{\MM}=\He$, and this is the reason why sometimes the Hellinger
distance is called Fisher-Rao distance. However, it is better to distinguished
the refined concept of Fisher-Rao distances $\FR_\calS$ which depends on the
chosen submanifold or subset $\calS$ of the set of all measures $\MM$.
 
In general, one has $\FR_\calS(\mu_0,\mu_1) \geq \He(\mu_0,\mu_1) $ where
equality holds only in exceptional cases, namely, if the Hellinger geodesic, as
given in \eqref{eq:I.HellGeod}, totally lies in $\calS$. In case that $\calS$
is a smooth submanifold the local Fisher-Riemann metric is simply the
restriction of the quadratic form \eqref{eq:I.QuadrForm} of Hellinger type,
one can expect local closeness of $\He$ and $\FR_\calS$, i.e.\
\[
\He(\mu_0,\mu_1)^2 \leq \FR_\calS(\mu_0,\mu_1)^2 \leq \He(\mu_0,\mu_1)^2 + O\big(
\He(\mu_0,\mu_1)^3\big)  \quad \text{ as } \He(\mu_0,\mu_1)\to 0.
\]
Already the restriction of $\He$ to the probability measures in $\PP$ leads to
a new distance, namely the Bhattacharya distance \cite{Bhat42DD,Rao45IAAE}
\[
\Bh(\mu_0,\mu_1) = 2\sigma \arcsin \big(\tdfrac1{2\sigma}\,\He(
\mu_0,\mu_1)\big) . 
\]
We refer to Section \ref{su:BhattDist} for more details and emphasize that
$\sigma>0$ appears nonlinearly. 

The plan of the paper is as follows. In Section \ref{se:Hellinger} we present
basic properties of the Hellinger distance such its the geodesic curves, the
embedding property into a Hilbert spaces (showing that the geometry is
flat), and the behavior under pushforwards. Moreover, we present some
historical remarks about Hellinger's contribution and the development of 
the name ``Hellinger distance''. The only  mathematically new part of these
notes are the short and self-contained characterization of absolutely
continuous curves in $(\MM,\He)$ using metric speed and the growth equation,
see Theorem \ref{th:AbsContCurve}. 

In Section 3 we discuss the abstract definition of the Fisher-Rao distance for
general subsets $\calS\subset \MM$. After treating the most important example
$\calS=\PP$ leading to the Bhattacharya distance, we show how the Fisher-Rao
distance on a set $\calP \subset \PP$ can be used to construct the Fisher-Rao
distance on the cone $\calS={[0,\infty[}\calP\subset \MM$ in Theorem
\ref{th:FRCones}. This is a general construction (see e.g.\ \cite{BuBuIv01CMG} 
for general geodesic spaces, where the involvement of the scaling parameter
$\sigma>0$ is nontrivial. In Section \ref{su:ProdMeas} we show that the
Fisher-Rao distance for product probability measures satisfies
$\FR_{\calP_1\oti \calP_2}^2  = \FR_{\calP_1}^2 + \FR_{\calP_2}^2$. 

Section \ref{se:FamProbDistr} is devoted to simple examples, namely
$\calS_\mafo{trans}$ containing all translations of a measure on $\R^n$, the
family $\calS_\mafo{Poiss}$ of multivariate Poisson distributions on $\N_0^d$,
and the family $\calS_\mafo{exp}$ of exponential distributions on $(\calR_\geq
0)^n$. Finally, in Section \ref{su:Gaussian} we discuss the known results on
the Fisher-Rao distance on $\calS_\mafo{Gauss}$, the set of Gaussian
distributions on $\R^d$: only for $d=1$ an explicit formula is known, and for
$d\geq 2$ only partial results are available.

\section{Properties of the Hellinger distance}
\label{se:Hellinger}

As in \cite{Kaku48EIPM} we start from a measure space $(\Omega,\mathfrak A)$,
i.e.\ $\mathfrak A$ is a $\sigma$-algebra over the set $\Omega$. By $\MM$ we
denote the set of all (non-negative) finite measures on $(\Omega,\mathfrak A)$,
i.e.\ countably additive set functions. The subset of probability measures is
denoted by $\PP=\bigset{\mu\in \MM }{ \mu(\Omega)=1}$.  
 
\subsection{Hellinger's integral} 
\label{su:HellIntegral}

As mentioned above the Hellinger distance relies on the the so-called Hellinger
integral,  which is the geometric mean of two measures $\mu_0,\mu_1 \in
\MM$. In modern terms the measure $\sqrt{\mu_0\mu_1} \in \MM$ is
defined by the Radon-Nikodym derivative via 
\begin{equation}
  \label{eq:HellDistRaNy}
\sqrt{\mu_0\mu_1} = \Big(\frac{\rmd \mu_0}{\rmd \lambda}\Big)^{1/2} 
 \Big(\frac{\rmd\mu_1}{\rmd \lambda} \Big)^{1/2} \lambda  \quad \text{for
   every }\lambda \gg \mu_0,\mu_1. 
\end{equation}
The geometric mean can also be defined by partitions as follows 
\begin{equation}
  \label{eq:HellDistInfim}
  \sqrt{\mu_0\mu_1} (A) = \inf\Bigset{ \sum_{i\in I} \mu_0(A_i)^{1/2}
  \mu_1(A_i)^{1/2} }{ A= \bigcup_{i\in I} A_i, \ A_i\cap A_j =\emptyset \text{ for }
  i\neq j}.
\end{equation}
Using that $(r,s) \mapsto (rs)^{1/2}$ is concave, it is easy to see that
refining partitions of a set $A$ leads to a smaller value (use
$\sqrt{\theta_i{+}\theta_j}\sqrt{\nu_i{+}\nu_j} \geq
\sqrt{\theta_i}\sqrt{\nu_i} + \sqrt{\theta_j}\sqrt{\nu_j}$). The historical
 Section \ref{su:History} will explain how this construction is related to
 Hellinger's work in \cite{Hell07OQFU,Hell09NBTQ}. 

\begin{remark}[Kolmogorov and Hellinger integrals]
\label{re:KolmHellInteg}
More generally, for a
positively one-homogeneous concave function $\varphi:{[0,\infty[}^N\to
{[0,\infty[}$ and measures
$\mu_1,..,\mu_n \in \MM$ the measure $\varphi(\mu_1,\mu_1,..,\mu_n)
\in \MM$ can be defined by an infimum over partitions as in
\eqref{eq:HellDistInfim}. The Kolmogorov integral of $g \in
\rmL^\infty(\Omega)$ is then defined as $\int_\Omega g \dd
\varphi(\mu_1,\mu_1,..,\mu_n)$, see \cite{Kolm30UI}. In particular, using  
$\phi_\alpha (r,s)= r^\alpha s^{1-\alpha}$ with $\alpha\in [0,1]$ and two
measures $\mu_0, \mu_1 \in \MM$, one can define the measures 
$\phi_\alpha(\mu_0,\mu_1) \in \MM$ and the so-called $\alpha$-Hellinger integral
$\int_\Omega g \rmd \phi_\alpha(\mu_0,\mu_1)$ for $g\in \rmL^\infty(\Omega)$. 
\end{remark}

\subsection{The topology of $(\MM,\He)$}
\label{su:Topology}

The topology induced by $\He$ on $\MM$ is the norm topology induced by the
total variation
\[
\|\mu_1{-}\mu_0 \|_\mafo{TV} = \int_\Omega \big| \frac{\rmd \mu_1}{\rmd \lambda}
- \frac{\rmd \mu_0}{\rmd \lambda} \big| \dd \lambda \quad \text{for }
\mu_0,\mu_1\dd \lambda.
\]
However, we see that the total variation norm scales one-homogeneous with the
mass, while the Hellinger distance scales homogeneous of degree $1/2$: 
\[
\|r\mu_1{-}r\mu_0 \|_\mafo{TV} = r\|\mu_1{-}\mu_0 \|_\mafo{TV} 
\quad \text{and} \quad 
\He(r\mu_0,r\mu_1) = r^{1/2}\,\He(\mu_0,\mu_1)
\]
for $r\geq 0$ and $\mu_0,\mu_1\in \MM$. This is reflected in the lower and
upper estimate of the total variation norm, namely 
\begin{subequations}
  \label{eq:NormEstim}
\begin{align}
  \label{eq:NormEstim.low} 
& \|\mu_1{-}\mu_0 \|_\mafo{TV} \geq \frac1{\sigma^2}\,\He(\mu_0,\mu_1)^2,  
\\
 \label{eq:NormEstim.upp}
&\|\mu_1{-}\mu_0 \|_\mafo{TV} \leq
\sqrt{2(\mu_0(\Omega){+}\mu_1(\Omega))} \: \frac1\sigma\,\He(\mu_0,\mu_1).
\end{align}
\end{subequations}
To see this, choose $\lambda=\mu_0{+}\mu_1$ and write
$\mu_0=\big(\frac12-x\big)\lambda$ and $\mu_1=\big(\frac12+x)\lambda$ with
$x(\omega)\in [-\frac12,\frac12]$ \ $\lambda$-a.e.\  Using the elementary
estimates
\begin{equation}
  \label{eq:ElemEstim}
  2x^2\leq 1-\sqrt{1{-}4x^2} \leq 2|x| \quad \text{for } |x|\leq 1/2,
\end{equation}
 the first estimate
follows from
\[
\|\mu_1{-}\mu_0 \|_\mafo{TV} = \int_\Omega 2|x|\dd \lambda \geq  \int_\Omega\! 
\big(1{-}\sqrt{1{-}4x^2}\big)\dd \lambda = \frac1{\sigma^2}\,\He(\mu_0,\mu_1)^2. 
\]
The second estimate follows via an application of Cauchy-Schwarz' estimate:
\begin{align*}
\|\mu_1{-}\mu_0 \|_\mafo{TV}^2&= \bigg(\int_\Omega\! 2|x|\dd \lambda\bigg)^2
\leq  \int_\Omega \! 2\dd \lambda\int_\Omega\! 2x^2\dd \lambda\\
& \leq 2 \big(\mu_0(\Omega){+}\mu_1(\Omega)\big) \int_\Omega\! 
 \big(1- \sqrt{1{-}4x^2}\,\big)\dd \lambda = 
\frac2{\sigma^2}\big(\mu_0(\Omega){+}\mu_1(\Omega)\big)\,\He(\mu_0,\mu_1)^2 . 
\end{align*}

\subsection{Absolutely continuous curves, metric speed and the continuity
  equation in $(\MM,\He)$} 
\label{su:AbsCurves}

In this subsection we use the abstract theory developed in
\cite[Sec.\,1.1]{AmGiSa05GFMS} for a general metric space $(M,\calD)$. For
$p\in [1,\infty]$, a curve $\gamma:[0,1]\to M$ is called is $p$-absolutely
continuous if there exists $g\in \rmL^p([0,1])$ such that
$\calD(\gamma(r),\gamma(t)) \leq \int_r^t g(s) \dd s$ for all
$0\leq r< t \leq 1$. Every rectifiable curve can be reparametrized to a
Lipschitz curve (i.e.\ $p=\infty$), so we see that the optimal $p$ depends on
the parametrization.

In \cite[Thm.\,1.1.2]{AmGiSa05GFMS} it is shown that the metric speed 
\[
\lim_{h\searrow  0} \frac1h \,\calD\big(\gamma(s),\gamma(s{+}h)\big) :=
|\dot\gamma|_\calD(s) 
\]
exists a.e.\ in $[0,1]$, and that $\calD(\gamma(r),\gamma(t)) \leq \int_r^t
|\dot\gamma|_\calD(s) \dd s $. 

The aim of this subsection  is to characterize the metric speed and relate it
properly to Hellinger's quadratic form. The theory is developed in analogy to
\cite[Ch.\,8]{AmGiSa05GFMS}, however our case of the Hellinger distance is
considerably simpler than the case of the Otto-Wasserstein theory developed
there. The new ingredient is the so-called generalized continuity equation,
which should rather be called a \emph{growth equation} (GE) here. Given a curve
$\mu:[s_0,s_1] \to \MM; s \mapsto \mu_s$ we define the measure $\mu_{[0,1]}$ on
$Q:=[0,1]\ti \Omega$ via 
\[
\int_Q h(s,\omega) \dd\mu_{[0,1]}(s,\omega) := \int_{[0,1]}\! \int_\Omega
h(s,\omega) \dd\mu_s(\omega) \dd s .
\]
For a  growth-rate function $\xi \in \rmL^1(Q;\mu_{[0,1]})$ we say that the
pair $(\mu,\xi)   $ is a weak solution of the growth equation $\pl_s \mu_s =
\xi_s \mu_s$ if 
\begin{equation}
\label{eq:GrowthEqn}
  \int_Q \big( \eta'(s) \bm1_A(\omega)  + \xi(s,\omega) \eta(s)\bm1_A (\omega)\big)
  \dd \mu_{[0,1]}(s,\omega) \quad \text{for all }\eta\in
  \rmC_\rmc^1({]s_0,s_1[}) \text{ and } A \in \mathfrak A. 
\end{equation}

In some respects, the present theory is much simpler than the corresponding
theory for the Otto-Wasserstein case developed in \cite[Ch.\,8]{AmGiSa05GFMS};
however that are new complications because of the change of support of the
measures. From the equation $\pl_s \mu_s = \xi_s \mu_s$ one would naively guess
that the solution can be written as
$\mu(s) = \exp\big( \int_0^s\xi(r,\cdot)\dd r\big) \mu(0)$, but that cannot be
true in case where $\mu(0)$ and $\mu(1)$ have different support. For instance,
choosing $\omega_0\neq \omega_1$ in $\Omega $ and considering the curve
\[
  \mu(s) = a_0(s) \delta_{\omega_0} + \theta a_1(s) \delta_{\omega_2} \quad
  \text{with } \theta\geq 0,\ a_0(s) = (1{-}s)^{\gamma_0} \text{ and }a_1(s) =
  s^{\gamma_1},
\]
we have $\mu(j)=\theta^j \delta_{\omega_j}$ (Dirac measure) for
$j=1,2$. Moreover, the growth equation is satisfied by $(\mu,\xi)$ if
$\xi(s,\omega_j)= a'_j(s)/a_j(s)$. Moreover, for
$\min\{\gamma_0,\gamma_1\} >1{+}p$ we have
\[
\int_Q |\xi |^p \dd \mu_I= \int_0^1 \sum_{j=0}^1 \theta^j |a'_j|^p/a_j^{p-1}
\dd s = \sum_{j=0}^1 \theta^j \frac{\gamma_j^p}{\gamma_j{-}1{-}p} < \infty.
\]  
We also note that for the given growth-rate function $\xi$ and the given
initial condition $\mu(0)=\delta_{\omega}$ there are infinitely many solution
pairs $(\mu,\xi)$ for the growth equation \eqref{eq:GrowthEqn}, because
$\theta\geq 0$ is arbitrary. 

For simplicity, we restrict to the natural case $p=2$.  

\begin{theorem}
\label{th:AbsContCurve}
(A) If $\mu:[s_0,s_1]\to \MM$ is $2$-absolutely continuous in $(\MM,\He)$, then
there exists $\xi \in \rmL^2([s_0,s_1]\ti \Omega)$ such that $(\mu,\xi)$ solve
the growth equation \eqref{eq:GrowthEqn} and and the metric speed satisfies
\begin{equation}
  \label{eq:SpeedChar}
  |\dot \mu|_\He(s) = \frac\sigma2 \|\xi_s\|_{\rmL^2(\Omega,\mu_s} =\frac\sigma2
  \bigg( \int_\Omega \xi(s,\omega)^2\dd \mu_s(\omega)\bigg)^{1/2} \quad
  \text{a.e. on } [0,1].
\end{equation}

(B) Vice versa, if $\mu :[s_0,s_1]\to \MM$ is a continuous curve,
$\xi\in \rmL^2([s_0,s_1]\ti \Omega,\mu_{[s_0,s_1]})$, and $(\mu,\xi)$ solves
\eqref{eq:GrowthEqn}, then $\mu$ is $2$-absolutely continuous in $(\MM,\He)$
and \eqref{eq:SpeedChar} holds. 
\end{theorem}
Before going into the proof of the result we emphasize that relation
\eqref{eq:SpeedChar} features Hellinger's quadratic form
\eqref{eq:I.QuadrForm}. From (the weak form 
of) the) growth equation $\pl_s \mu_s=\xi_s \mu_s$  we have 
\[
\xi_s = \frac{\rmd (\pl_s \mu_s)}{\dd \mu_s} \ \ \text{ and \eqref{eq:SpeedChar}
  means } \ \ \big(|\dot\mu|_\He(s)\big)^2 = \frac{\sigma^2}4 \int_\Omega \Big(
\frac{\rmd (\pl_s \mu_s)}{\dd \mu_s} \Big)^2\dd\mu_s= \bfg_{\mu_s}\big(\pl_s
\mu_s, \pl_s \mu_s \big) .
\]
\begin{proof} To simplify notation we only consider the case $[s_0,s_1]=[0,1]=:I$. 

\underline{\em Proof of part (A):} We proceed in analogy to
\cite[Thm.\,8.3.1]{AmGiSa05GFMS}.  

For simplicity, we set $\mathscr V:=\rmL^2(Q;\mu_I)$ where $Q=I\ti \Omega$ and
$\mu_I=\mu_{[0,1]}$. Moreover, we define the dense subset 
\[
V=  \bigset{ (t,\omega) \mapsto \sum_{i=1}^{n} \eta_i(t) \bm1_{A_i}(\omega)} 
    {n\in \N,\ \eta_i\in \rmC^1(I),\ \eta_i(0)=0=\eta_i(1), \ A_i \in \mathfrak A}.
\]
On $V$ we define the linear mapping $L:V \to \R$ via 
\[
L\varphi = -\int_Q \pl_s \varphi(s,\omega) \dd \mu_I(s,\omega) = -\int_I \int_\Omega
\pl_s\varphi(s,\omega) \dd \mu_s(\omega) \dd s. 
\]
We now want to show that $L$ can be extended continuously  on all of $\mathscr
V$. We define $S\subset [0,1]$ to be the set of those $s$ where
$|\dot\mu|_\He(s)$ exist. Then, for all measurable and bounded $g:\Omega
\to \R$ we have 
\[
\big(\langle g,\mu_{s+h}\rangle - \langle g,\mu_s\rangle\big) = 
\int_\Omega g (1{-}2\theta_h) \dd \lambda_h \quad \text{where
}\lambda_h=\mu_{s+h} + \mu_s, 
\]
$\gamma_s = \theta_h \lambda_h$,  and $\gamma_{s+h}=(1{-}\theta_h)\lambda_h$.  
Using Cauchy-Schwarz' estimate and \eqref{eq:ElemEstim} we find 
\[
\frac1h \big|\langle g,\mu_{s+h}\rangle - \langle g,\mu_s\rangle\big| \leq 
\big\| g\|_{\rmL^2(\Omega,\lambda_h) }  \frac{\sqrt2}\sigma 
\frac{\He(\mu_s,\mu_{s+h})}{h} .
\]
We now assume $s\in S$ and use that $\lambda_h \to 2 \mu_s$ (by strong
continuity of $t \mapsto \mu_t$. Thus we have 
\begin{equation}
  \label{eq:Eqn123}
  \limsup_{h\searrow 0} \frac1h \big|\langle g,\mu_{s+h}\rangle - \langle
g,\mu_s\rangle\big| \leq    \big\| g\|_{\rmL^2(\Omega,\mu_s)} \: 
\frac2\sigma |\dot\mu|_\He(s). 
\end{equation}
Now we consider a general $\varphi \in V$ and extend it (continuously!) by
$0$. Moreover, extend $s\mapsto \mu(s)  $ by $\mu(1)$ for $ s\geq 1$. 
Then, we have   
\begin{align*}
\int_Q  \pl_s \varphi \dd \mu_I & = \lim_{h\searrow 0} \int_Q
\frac1h\big(\varphi(s,\omega) {-} \varphi(s{-}h,,\omega)\big) \dd \mu_I 
\\
&= \lim_{h\searrow 0}\bigg(\int_I\frac{\langle \varphi_s,\mu_s\rangle {-} \langle
\varphi_s,\mu_{s+h} \rangle}h  \dd s - \frac1h\int_0^h \!\!\langle
\varphi_{s-h},\mu_s\rangle \dd s + \frac1h \int_{1-h}^1 \!\!\langle
\varphi_s,\mu_{s+h}\rangle \bigg) .
\end{align*}
Because $\varphi_0=0=\varphi_1$ the last two terms vanish with $h \searrow 0$. Hence,
together with \eqref{eq:Eqn123}, we find 
\begin{equation}
  \label{eq:L.bounded}
  \big| L(\varphi)\big| = \bigg|\int_Q  \pl_s \varphi \dd \mu_I \bigg| \leq
\int_I \big\| \varphi_s \big\|_{\rmL^2(\Omega,\mu_s)} \frac2\sigma
|\dot\mu|_\He(s) \dd s . 
\end{equation}
By assumption $s\mapsto |\dot\mu|_\He(s)$ lies in $\rmL^2(I)$, hence $L$ can be
extended continuously to $\mathscr V$. By Riesz' representation theorem for the
Hilbert space $\mathscr V$, there
exists $\xi\in \mathscr V$ such that $L(\varphi)= \int_Q \xi \varphi \dd
\mu_I$, but this shows that $(\mu,\xi)$ solve the growth equation
\eqref{eq:GrowthEqn}. 

Moreover, take $\eta=\bm1_{[s_0,s_1]}$, then using \eqref{eq:L.bounded} we find 
\begin{align*}
\int_{s_0}^{s_1}\!  \int_\Omega \xi_s^2\dd \mu_s \dd s &=\int_Q \eta \xi^2
\dd\mu_I = L(\eta\xi) 
\overset{\text{\eqref{eq:L.bounded}}}{\leq}  \frac2\sigma \int_Q \eta
\|\xi_s\|_{\rmL^2(\Omega,\mu_s)} |\dot\mu|_\He(s) \dd s
\\
& \leq \frac2\sigma  \bigg(
\int_{s_0}^{s_1}\!  \int_\Omega \xi_s^2\dd \mu_s \dd s \bigg)^{1/2}
 \bigg( \int_{s_0}^{s_1}\!\big( |\dot\mu|_\He(s) \big)^2\dd s \bigg)^{1/2}.
\end{align*}
Since $s_0$ and $s_1$ with $0\leq s_0 < s_1\leq 1$ are arbitrary we conclude 
\[
\int_\Omega \xi_s^2 \dd \mu_s \leq \frac4{\sigma^2}   \big( |\dot\mu|_\He(s)
\big)^2 \quad \text{for a.a. } s \in [0,1]. 
\]
Thus, we have established \eqref{eq:SpeedChar} with ``$\geq$'' instead of
``$=$''.   The opposite inequality will be shown via part (B).\medskip

\underline{\em Proof of part (B):}  The measure $\mu_I \in \mathfrak (Q)$ has
two disintegration with respect to $Q=I\ti \Omega$, namely into $\dd \mu_i =
\dd\mu_s(\omega) \dd s$ and $\dd \mu_I= \dd \nu_\omega(s)\dd \ol\mu(\omega)$,
where $\ol\mu\in \MM$ and $\nu_\omega \in \mathfrak P([0,1])$ for
$\ol\mu$-a.a.\ $\omega \in \Omega$.  Here $\ol\mu(A)=\mu_I(I\ti A)$ or $\ol\mu=
\int_0^1 \gamma_s\dd s$ (Bochner integral). From $\xi \in \rmL^2(Q;\mu_I)$ we
have $\xi_\omega:=\xi(\cdot,\omega) \in \rmL^2([0,1],\nu_\omega)$ \,$\ol\mu$-a.e.\ in
$\Omega$. Testing weak growth equation \eqref{eq:GrowthEqn} with
$\varphi(s,\omega) = \eta(s) \bm1_A(\omega)$ we find 
\[
\int_A \int_I \big( \eta'(s) + \xi(s,\omega) \dd \nu_\omega(s) \dd
\ol\mu(\omega) = 0  \quad \text{for all } \eta\in \rmC^1_0([0,1]) 
\text{ and } A \in \mathfrak A. 
\]
As $A\in \mathfrak A$ is arbitrary, we conclude 
\[
\Big( \: \forall\,\eta\in \rmC^1_0([0,1]): \  \int_I \big( \eta'(s) +
\xi_\omega(s) \big) \dd \nu_\omega(s) =0 \: \Big)\quad
 \ol\mu\text{-a.e.\ in } \Omega.
\]
Thus, we have reduced the problem in $\MM$ to a pointwise problem scalar
problem.

  
From $\xi_\omega \in \rmL^2([0,1],\nu_\omega)$ we see that
$\zeta_\omega: =\xi_\omega \nu_\omega $ is a signed measure on $[0,1]$ and we
have $\pl_s \nu_\omega=\zeta_\omega$ in the distributional sense. Hence,
$\nu_\omega$ is absolutely continuous with respect to $\dd s$ with
$\nu_\omega=n_\omega \dd s $ and $n_\omega \in \rmB\rmV([0,1])$. Inserting this
once again into the weak equation we find $n_\omega\in \rmW^{1,1}([0,1])$ with
$n'_\omega(s)= \xi_\omega(s) n_\omega(s)$ a.e.\ in $[0,1]$.

Omitting the subscript for the moment, we set $h(s)=\sqrt{n(s)}\geq 0$ and find 
$2h h'=\xi h^2$, which implies either $h=0$ or $h'=\frac12\xi h$. Since $h(s)=0$
and $h\geq 0$ imply $h'(s)=0$ a.e., we obtain 
\[
h'(s)^2 = \frac14\xi(s)^2 h(s)=  \frac14\xi(s)^2 n(s)
 \quad \text{a.e. in }[0,1]. 
\]
Thus, we have 
\begin{equation}
  \label{eq:350}
\begin{aligned}
\big( \sqrt{n_\omega(1)} - \sqrt{n_\omega(0)}\big)^2& =  \big(h(1){-}
h(0)\big)^2 = \bigg(\int_I h'(s)\dd s\bigg)^2\\
& \leq \int_I\big( h'(s)\big)^2\dd s 
  = \int_I \frac14\xi_\omega(s)^2 n_\omega(s) \dd s   
\end{aligned}
\end{equation}
Noting that $\dd\mu_s=n_\omega(s)\ol\mu$ we can integrate the this estimate and
arrive at
\begin{align}
\label{eq:355}
\begin{aligned}
\He(\mu_0,\mu_1)^2 &= \sigma^2\int_\Omega (\sqrt{n_\omega(1)}- 
\sqrt{n_\omega(0)}\big)^2 \dd \ol\mu(\omega) 
 \leq \frac{\sigma^2}4 \int_\Omega \int_I \xi_\omega^2 n_\omega\dd s \dd \ol\mu 
\\
& = \frac{\sigma^2}4\int_I g(s)^2  \dd s \quad \text{with } 
  g(s): = \| \xi_s\|_{\rmL^2(\Omega,\gamma_s)} .
\end{aligned}
\end{align}

The same estimate can be done on each subinterval $[r,t]\subset [0,T] $ giving 
\begin{align*}
\He(\mu_r,\mu_t) &\leq \frac{\sigma}2\:(t{-}r) \bigg(\frac1{t{-}r} \int_{[r,t]}g(s)^2 \dd
 s\bigg)^{1/2}. 
\end{align*}
where the factors $(t{-}r)$ disappear because the interval $[0,1]$ in 
\eqref{eq:355} needs to be rescaled to $[r,t]$. Defining the partition points
$s_i=r+i(t{-}r)/N$ for $i=0,1,\ldots,N$ and the piecewise constant function  
\[
G_N(s)= \sum_{i=1}^N \frac1{s_i{-}s_{i-1}} \int_{s_{i-1}}^{s_i} g(r)^2 \dd r
\,\bm1_{[s_{i-1}, s_i]}(s) 
\]
we find $G_N\to g^2$ in $\rmL^2([r,t])$ and $\sqrt{G_N}\to g$ in
$\rmL^2([0,1])$. With this we have 
\begin{align*}
\He(\mu_r,\mu_t) & \leq \sum_{i=1}^N \He(\mu_{s_{i-1}}, \mu_{s_i}) 
    \leq\frac{\sigma}2\: \sum_{i=1}^N  (s_{i}{-}s_{i-1})\sqrt{ G_N(s_{i-1/2})}
\\
& = \frac{\sigma}2 \int_r^t \sqrt{G_N(s)}
\dd s \to \frac{\sigma}2\int_r^t g(s) \dd s \text{ for } N\to \infty. 
\end{align*}
This shows that $\mu:[0,1]\to \MM$ is $2$-absolutely continuous and
$|\dot\mu|_\He(s) \leq \frac{\sigma}2\,g(s)$ a.e.\ in $[0,1]$, which is the
``$\leq $'' part of \eqref{eq:SpeedChar}. 

With this Theorem \ref{th:AbsContCurve} is established. 
\end{proof}

\subsection{Geodesics curves}

According to \cite{Kaku48EIPM} (using the choice $\sigma=1$) the Hellinger
distance is given by  
\[
\He(\mu_0,\mu_1)^2 = \sigma^2\big(\mu_0(\Omega) + \mu_1(\Omega) - 2
\sqrt{\mu_0\mu_1}(\Omega)\big). 
\]
One importance starting point of this paper is that this distance is a geodesic
distance. For each pair $(\mu_0,\mu_1) \in \MM^2$ there exists a unique
constant-speed geodesic (simply called geodesic in the sequel), i.e.\ a curve
$\gamma:[0,1] \to \MM$ such that
\[
\gamma(0)=\mu_0,\ \ \gamma(1)=\mu_1, \ \ 
\He\big(\gamma(s),\gamma(t)\big) = |s{-}t|\,
\He (\mu_0,\mu_1) \text{ for all } s,t\in [0,1]. 
\]
This geodesic is given by 
\begin{align*}
\gamma^\He_{\mu_0\to\mu_1}(s)
&= \Big( \big( (1{-}s)\big(\tdfrac{\rmd \mu_0}{\rmd \lambda}\big)^{1/2} + 
s \,\big(\tdfrac{\rmd \mu_1}{\rmd \lambda}\big)^{1/2}\Big)^2 \lambda 
\\
& = (1{-}s)^2\mu_0 + s^2 \mu_1 +2(s{-}s^2)
\sqrt{\mu_0\mu_1} \\
& = (1{-}s)\mu_0 + s \mu_1 -(s{-}s^2) \tdfrac1{\sigma^2} \He(\mu_0,\mu_1)^2
   \ \text{ for } s\in [0,1],
\end{align*}
where $\lambda \in \MM$ is arbitrary as long as $\mu_0,\mu_1 \ll \lambda$. 

The growth equation along a geodesic can also be given explicitly, namely
\begin{subequations}
\label{eq:Geod} 
\begin{equation}
  \label{eq:Geod.a}
  \pl_s \gamma^\He_{\mu_0\to\mu_1}(s) = \xi(s,\cdot)
   \gamma^\He_{\mu_0\to\mu_1}(s) \quad  \text{with } \xi(s,\omega)
   = \frac{2\,(f_1(\omega)-f_0(\omega)}{ (1{-}s)f_0(\omega)+ s f_1(\omega) } ,
\end{equation}
where $f_j= \big(\dd \mu_j/\dd \lambda \big)^{1/2}$. Note that $\xi$ satisfies
the equation
\begin{equation}
  \label{eq:Geod.b}
  \pl_s \xi(s,\omega) +\frac12 \xi(s,\omega)^2=0.  
\end{equation}
\end{subequations}
which is completely independent of $\mu_0$ and $\mu_1$. The system
\eqref{eq:Geod} form the geodesic equations, which are a special case of the
equations for the Hellinger-Kantorovich geodesics derived in
\cite[Eqn.\,(5.1)]{LiMiSa16OTCR} and \cite[Eqn.\,(8.72)]{LiMiSa18OETP}. 

An important feature is that along the geodesics the total mass of the measure
is exactly quadratic, namely 
\begin{equation}
  \label{eq:MassQuad}
\begin{aligned}  \gamma^\He_{\mu_0\to\mu_1}(s)(\Omega) &
= (1{-}s)^2 \mu_0(s) + s^2 \mu_1(\Omega) +
  2(s{-}s^2)\sqrt{\mu_0\mu_1}(\Omega)\\
& =  (1{-}s) \mu_0(\Omega) + s
  \mu_1(\Omega) -(s{-}s^2) \tdfrac1{\sigma^2} \He(\mu_0,\mu_1)^2. 
\end{aligned}
\end{equation}
In particular, one can define a Hellinger average by taking the midpoint of the
geodesics:
\begin{equation}
  \label{eq:HellMidpoint}
  A^\He(\mu_0,\mu_1) := \gamma^\He_{\mu_0\to\mu_1}(1/2) = \frac14 \mu_0 +
  \frac14\mu_1 + \frac12 \sqrt{\mu_0\mu_1}. 
\end{equation}

\subsection{Properties similar to Hilbert-space geometry} 
\label{su:HellHilbert}

The last form of the geodesic already indicates that $(\MM,\He)$ is somehow
related to the positive cone in the Hilbert space
$\rmL^2(\Omega,\lambda)$. However, here the measure $\lambda$ depends on the
measures $\mu_j$ that are relevant for the current construction. This embedding
is already included in \cite[Sec.\,4]{Kaku48EIPM}, see Remark \ref{re:KakuEmbed}.

(H1) We have a counterpart to the \emph{parallelogram identity} for all
$\mu_0,\mu_1,\mu_2\in \MM$: 
\begin{equation}
  \label{eq:H1Parallel}
2 \He(\mu_0,\mu_1)^2 + 2
\He(\mu_0,\mu_2)^2 = \He(\mu_1,\mu_2)^2 + 4 \,\He\big( \mu_0,
A^\He(\mu_1,\mu_2)\big)^2,
\end{equation}
where $A$ denotes the average defined in \eqref{eq:HellMidpoint}. See Figure
\ref{fig:Parallel} for a visualization. 
\begin{figure}[h]
\begin{minipage}{0.45\textwidth}
\begin{tikzpicture}
\draw[ultra thick] (0,0).. controls (2,0.3) .. (4.5,0);
\draw[ultra thick] (0,0).. controls (0.6,1) .. (1.4,2);
\draw[ultra thick] (0,0).. controls (1.6,0.6) .. (3.2,1.1);
\draw[ultra thick] (4.5,0).. controls (3.5,1) .. (1.4,2);
\draw[thick, dashed] (0,0).. controls (3.8,1.4) .. (7,1.6);
\draw[thick, dashed] (4.5,0).. controls (5.8,1) .. (7,1.6);
\draw[thick, dashed] (1.4,2).. controls (4.2,2) .. (7,1.6);
\fill (0,0) circle (0.3em) node[below] {$\mu_0^{^{}}$}; 
\fill (4.5,0) circle (0.3em) node[below] {$\mu_1^{^{}}$}; 
\fill (1.4,2) circle (0.3em) node[left] {$\mu_2\;$}; 
\fill (3.2,1.1) circle (0.3em) node[above] {$\mu_A$\rule[-0.4em]{0em}{0em}}; 
\end{tikzpicture}
\end{minipage}
\hfill
\begin{minipage}{0.5\textwidth}
\caption{A visualization of the parallelogram identity in $(\MM,\He)$ where all
  edges (full lines) are geodesics, while the broken lines contain curves that
  may lead to signed measures lying in $\rmL^2(\Omega,\lambda)$. The center of
  the parallelogram is $\mu_A=A^\He(\mu_1,\mu_2)$.}
\label{fig:Parallel}
\end{minipage}
\end{figure}
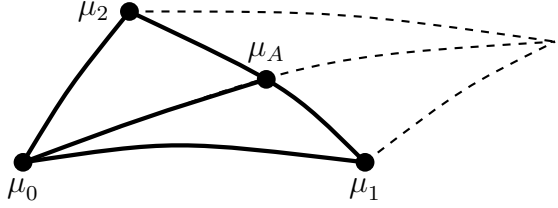 

(H2) A second instance occurs when looking at the \emph{squared distance} along
geodesic curves. For every three points $\mu_0, \mu_1, \eta \in \MM$ we have
\emph{geodesic $2$-convexity} as well as \emph{geodesic $2$-concavity}, i.e.\ for all $ s\in
[0,1] $ we have 
\begin{equation}
  \label{eq:2CvxCcv}
\He\big(\gamma^\He_{\mu_0\to\mu_1}(s), \eta\big)^2 = (1{-}s) \He(\mu_0,\eta)^2 
 + s \He(\mu_1,\eta)^2 -2\, \frac{s{-}s^2}2  \He(\mu_0,\mu_1)^2 . 
\end{equation}

(H3) Finally, we may define \emph{angles} between geodesics emanating from a
point $\mu_0$. Setting $\gamma_j(s) = \gamma^\He_{\mu_0\to \mu_j}(s)$ for $j=1,2$ the
angle between geodesics is defined in the sense of geodesic spaces 
(see \cite{BuBuIv01CMG, LasMie19GPCA}) via 
\[
\sphericalangle(\gamma_1,\gamma_2):= \arccos \Big(\lim_{s,t\SEarrow0}
\frac{\He(\mu_0,\gamma_1(s))^2 {+} \He(\mu_0,\gamma_2(t))^2 {-}
  \He(\gamma_1(s),\gamma_2(t))^2} {2 \, \He(\mu_0,\gamma_1(s))
  \, \He(\mu_0,\gamma_2(t)) }\Big) \ \in \ [0,\pi]
\] 
whenever it exists. Exploiting the quadratic formula \eqref{eq:2CvxCcv} a
straightforward calculation shows that the fraction in the above definition is
indeed constant (as in Hilbert spaces) and we find 
\begin{equation}
  \label{eq:AngleHell}
  \sphericalangle(\gamma_1,\gamma_2) =  \arccos \Big( 
  \frac{\He(\mu_0,\mu_1)^2 {+} \He(\mu_0,\mu_2)^2 {-} \He(\mu_1,\mu_2)^2 }
   {2 \,\He(\mu_0,\mu_1)\, \He(\mu_0,\mu_2)}\Big).
\end{equation}
Clearly this is the same formula as in planar geometry, which is valid in all
Hilbert spaces.

\begin{remark}[Embedding into Euclidean space] \slshape
\label{re:KakuEmbed}
The observation that certain
  subsets of $\MM$ can be embedded into a Euclidean (Hilbert) space was
  crucial for the work in \cite{Kaku48EIPM}. There Section 4 is entitled
  ``\emph{Embedding into Euclidean space}''. Moreover, the introduction
contains the following text:
\begin{quote}\itshape
The results of this paper were much amplified and the arguments used below
much simplified, thanks to certain suggestions kindly made by Professor John
von Neumann. In particular, the introduction of inner product and isometric
embedding of $\MM$ into a general Euclidean space, as well as the indication
of relationship of this paper with earlier works of E. Hellinger, as are discussed
in \S4, are due to Professor J. von Neumann. For all these I wish to express my
heartiest thanks.
\end{quote} 
Indeed, given any $\lambda\in \MM$ the subspace $\rmL^1(\Omega,\lambda)\subset
\MM $ equipped with the Hellinger distance can be embedded isometrically into
the Hilbert space $\rmL^2(\Omega,\lambda)$ via the mapping $ g\lambda \mapsto
\sqrt{g}\,\lambda$ (using the normalization $\sigma=1$). 
\end{remark}
 
\subsection{Hellinger distance for product measures}
\label{su:HellProduct}

Already the paper  \cite{Kaku48EIPM} shows that the Hellinger integral is very
useful when studying product measures. Assuming $\Omega=\Omega_1\ti \Omega_2$
with associated $\sigma$-algebras we can consider measures $\nu_1\in
\MMM{1}$ and $\nu_2 \in \MMM{2}$ and define the product measure $\mu=\nu_1\oti
\nu_2$ as the unique measure $\mu\in \MM$ satisfying
\[
\nu_1\oti \nu_2\big( A_1\ti A_2) = \nu_1(A_1) \,\nu_2(A_2). 
\]
It now follows easily from the definition of the Hellinger integral that it is
compatible with the product structure in the sense that 
\begin{equation}
  \label{eq:HellIntProd}
  \sqrt{(\nu_1\oti \nu_2) \,(\eta_1\oti \eta_2)} = \sqrt{\mu_1\eta_1} \otimes 
  \sqrt{\nu_2\eta_2} .
\end{equation}
Indeed, this relation and its generalization to infinite product measures is
the basis of the analysis in \cite{Kaku48EIPM}.  

From this we can derive a corresponding formula for the Hellinger distance
between two product measures, namely 
\begin{equation}
  \label{eq:HellDistProd}
  \He\big(\nu_1\oti \nu_2, \eta_1\oti \eta_2\big)^2 = \sigma^2\big(
 \ol\nu_1\ol\nu_2 + \ol\eta_1\ol\eta_2 \big) -  \frac{\sigma^2}2 
\big(\ol\nu_1{+}\ol\eta_1{-}\tdfrac1{\sigma^2} \He(\nu_1,\eta_1)^2 \big) 
\big(\ol\nu_2{+}\ol\eta_2 {-} \tdfrac1{\sigma^2} \He(\nu_2,\eta_2)^2 \big) , 
\end{equation}
where we abbreviated $\ol\nu_j=\nu_j(\Omega_j)$ and
$\ol\eta_j=\eta_j(\Omega_j)$. We see that the Hellinger distance of the product
measures can be expressed solely in therms of the total masses of the
individual measures and the Hellinger distances between the corresponding
factors. 

The above result takes a much simpler form if we restrict to probability
measures $\nu_1,\eta_1\in \PPP1$ and $\nu_2,\eta_2\in \PPP2$:
\begin{equation}
  \label{eq:HellProbProd}
  \He\big(\nu_1\oti \nu_2, \eta_1\oti \eta_2\big)^2 = \He(\nu_1,\eta_1)^2 +
  \He(\nu_2,\eta_2)^2 -  \tdfrac1{2\sigma^2} 
                \He(\nu_1,\eta_1)^2\He(\nu_2,\eta_2)^2. 
\end{equation}

We already emphasize at this point that in general a geodesic curve between two
product measures does not stay a product measure any more. It is rather a
convex combination of three product measures, namely
\[
\gamma^\He_{\nu_0\oti\eta_0\to \nu_1\oti\eta_1}(s) = (1{-}s)^2 \nu_0\oti\eta_0
+ s^2  \nu_1\oti\eta_1 + 2(s{-}s^2) \sqrt{\nu_0\eta_0}\oti \sqrt{\nu_1\eta_1}. 
\]
To find the shortest connecting path consisting of product measures will be
treated as a special case of a Fisher-Rao distance in Section \eqref{su:ProdMeas}.

The case of product measures with more than two factors (as in
\cite{Kaku48EIPM}) works analogously: For 
$\nu_j,\eta_j \in \PPP{j}$ we have 
\begin{equation}
  \label{eq:MultiProducts}
\begin{aligned}
1 {-}\tdfrac1{2\sigma^2} \,\He\big( \mathop{\oti}_{i=1}^n \nu_i, 
     \mathop{\oti}_{j=1}^n \eta_j\big)^2
&  = \sqrt{(\oti \nu_i)  (\oti \eta_j)} \big( 
      \mathop{\ti}\limits_{k=1}^n\Omega_k\big)
\\
& = \mathop{\Pi}\limits_{k=1}^n \!\sqrt{\nu_k\eta_k}(\Omega_k) 
 = \mathop{\Pi}\limits_{k=1}^n \!\big( 1{-} \tdfrac1{2\sigma^2} 
\He(\nu_k,\eta_k)^2\big) . 
\end{aligned}
\end{equation}

\subsection{Invariance under pushforwards of the Hellinger distance}
\label{su:Pushforward}

Considering two measure spaces $(\Omega,\mathfrak A)$ and $(\Sigma,\mathfrak
B)$ and a measurable mapping $\Phi:\Omega\to \Sigma$, the pushforward $\Phi_\#
\mu \in \mathfrak M(\Sigma)$ of $\mu$ is defined via 
\[
\Phi_\#\mu (B) := \mu\big(\Phi^{-1}(B)\big) \ \text{ for all } B \in \mathfrak B
\text{ and }\mu \in \MM, 
\]
where $\Phi^{-1}(B):= \bigset{ x\in \Omega}{ \Phi(x) \in B} \subset \mathfrak A$
such that $\Phi$ does not need to be injective. See
\cite[Sec.\,5.2]{AmGiSa05GFMS} for this and further properties of pushforward
measures.  

Using the infimum characterization \eqref{eq:HellDistInfim}  of
$\sqrt{\mu_0\mu_1}$ and comparing the admissible partitions for the two sides 
we easily find
\[
   \Phi_\#\big(\sqrt{\mu_0\mu_1}\big) (B) \leq 
 \sqrt{(\Phi_\#\mu_0)(\Phi_\#\mu_1)}(B) \ \text{ for all } B \in \mathfrak B.
\]
Clearly, we have $\Phi_\# \mu_j(\Sigma)  =  \mu_j(\Omega)$, hence we
immediately find the monotonicity of the Hellinger distance under 
pushforward, with equality if the operation can be reversed. 

\begin{lemma}[Hellinger distance and pushforward]
\label{le:Pushforward}
For measure spaces $(\Omega,\mathfrak A)$ and $(\Sigma,\mathfrak
B)$ and a measurable mapping $\Phi:\Omega\to \Sigma$ we have 
\begin{equation}
  \label{eq:HePushforward}
  \forall\, \mu_0,\mu_1\in \MM\colon \
  \He_\Sigma\big(\Phi_\#\mu_0,\Phi_\#\mu_1 \big)  \leq \He_\Omega(\mu_0,\mu_1).
\end{equation}
If additionally $\Psi$ is one-to-one with measurable inverse $\Psi^{-1}$, then
\begin{equation}
  \label{eq:HePushforwInvert}
  \forall\, \mu_0,\mu_1\in \MM\colon \
  \He_\Sigma\big(\Psi_\#\mu_0,\Psi_\#\mu_1 \big) = \He_\Omega(\mu_0,\mu_1).
\end{equation}
\end{lemma}

In \cite{BaBrMi16UFRM} the case $\Sigma=\Omega$ being a smooth
finite-dimensional manifold without boundary and dimension $\geq 2$ is studied,
and it is shown that the Hellinger distance (called Fisher-Rao metric there)
restricted to probability measures is the only ``Riemannian distance'' that has
the invariance property \eqref{eq:HePushforwInvert}. There the theory is
restricted to smooth densities (instead of measures) and diffeomorphisms.

In \cite[last\,parag.]{BaBrMi16UFRM} it is shown that asking
\eqref{eq:HePushforwInvert} only for smooth diffeomorphisms on $\Omega = \bbS^1$ 
allows for more general distances $\calD$ in $\PP$ than multiplies of $\He$. 
Hence, it would be interesting to know whether enforcing
\eqref{eq:HePushforwInvert}  also for general measurable
homeomorphism rules out this pathology. More generally, one might conjecture
every distance $\calD:\PP\ti \PP \to {[0,\infty[}$ satisfying the invariance
\eqref{eq:HePushforwInvert} and the properties (H1), (H2), and (H3) in Section
\ref{su:HellHilbert} is a multiple of $\He$.

\subsection{Historical remarks} 
\label{su:History}

In his dissertation \cite{Hell07OQFU} and habilitation thesis
\cite{Hell09NBTQ}, Hellinger introduced integrals of the type $\int_a^b
u(t) \frac{\rmd f_1\rmd f_2}{\rmd g}$ for functions $u,f_1,f_2,g,h\in
\rmC^0([a,b])$ where additionally $g$ and $h$ are increasing and satisfy $\big(
f_j(t_2){-}f_j(t_1) \big)^2 \leq \big( g(t_2){-}g(t_1)\big) \big(
h(t_2){-}h(t_1)\big)$ for all $t_1,t_2$ with $a \leq t_1  < t_2  \leq b$.   
In modern terns using  the Radon-Nikod\'ym derivative, we would introduce a
dominating measure $\lambda \in \mathfrak ([a,b])$  and assume 
$\rmd f_j = \phi_j \dd \lambda$, $\rmd g = \gamma \dd \lambda$, and $\rmd
h = \eta \dd \lambda$ with the restriction $\phi_j^2 \leq \gamma\eta$. Then,
$\phi_1\phi_2/\gamma \leq \eta$ a.e.\ with respect to $\lambda$, and 
Hellinger's integral can be interpreted in the form 
\begin{equation}
  \label{eq:IntegralHelli}
   \int_a^b u(t) \frac{\rmd f_1\rmd f_2}{\rmd g} : = 
\int_a^b
u(t) \frac{\phi_1(t)\phi_2(t)}{\gamma(t) } \dd \lambda(t).
\end{equation}
However, as Hellinger's construction was much before the introduction of the
Radon-Nikod\'ym derivative, he used a convexity argument that is reminiscent to
the concavity in the definition of $\sqrt{\mu_1\mu_2}$ in
\eqref{eq:HellDistInfim}: Restricting the integral in \eqref{eq:IntegralHelli}
to the case $u\equiv 1$ and $f=f_1=f_2$ one shows that 
\[
\int_a^b \frac{\rmd f^2}{\rmd g} = \lim_{\Delta(\Pi)\to 0} \sum_{t_i\in \Pi} 
\frac{\big(f(t_i) {-} f(t_{i-1})\big)^2}{g(t_i){-}g(t_{i-1})} = 
\sup_{\Pi\in \mafo{Part}([a,b]}   \sum_{t_i\in \Pi} 
\frac{\big(f(t_i) {-} f(t_{i-1})\big)^2}{g(t_i){-}g(t_{i-1})},
\]
see \cite[\S4, p.\,234]{Hell09NBTQ}, because the discrete sum is increasing
under refinements of the partitions 
$\Pi$ of the interval $[a,b]$, by using the estimate 
\[
\frac{\big(f(t_i) {-} f(t_{i-2})\big)^2}{g(t_i){-}g(t_{i-2})}\leq
\frac{\big(f(t_i) {-} f(t_{i-1})\big)^2}{g(t_i){-}g(t_{i-1})} + 
\frac{\big(f(t_{i-1}) {-} f(t_{i-2})\big)^2}{g(t_{i-1}){-}g(t_{i-2})}.
\]
In \cite{Kolm30UI} the argument was generalized to arbitrary measure spaces, 
defining so-called Kolmogorov integrals by using the infimum construction, but
an explicit reference of Hellinger's work is given \cite[p.\,679]{Kolm30UI},
referring explicitly to \cite[p.234]{Hell09NBTQ}. 
  
Using the modern tool of the Radon-Nikod\'ym derivative,
\cite[Eqn.\,(11)]{Kaku48EIPM} introduces the so-called Hellinger integral
$\rho(\mu,\nu) = \int_\Omega \sqrt{\mu(\rmd\omega)\nu(\rmd\omega)}$ and defines
what is nowadays called the Hellinger distance on probability measures via
$\He(\mu,\nu) = \big(2{-} 2\rho(\mu,\nu)\big)^{1/2}$.

Since the early 1960s, the name Hellinger distance is consistently used in
probability theory and statistics (see. e.g.\ \cite{Leca70APAN}), which
was checked by a search of ``Hellinger distance'' in MathSciNet in 2023, which
led to more than 600 hits in abstracts or titles. In particular, Rao's paper
\cite[\S3, p.\,304]{RaoVar63DGP} introduces the Hellinger integral and the
Hellinger distance explicitly by name.   

The Fisher-Rao distance was popularized by \cite{Rao45IAAE} as geodesic
distance for the Fisher information metric. It is interesting to see that the
abstract version of the Fisher metric given in \eqref{eq:I.QuadrForm} is
exactly of the form of the Hellinger integrals \eqref{eq:IntegralHelli}
introduced already in \cite{Hell09NBTQ}, however in a rather restrictive
setting.

\section{Various Fisher-Rao distances}
\label{se:FisherRaoDist}

We first discuss the general construction of the Fisher-Rao distance
$\FR_\calS$ for general subset $\calS\subset \MM$ without direct reference to
the local Fisher information metric $\bfg_\mu(\nu_1,\nu_2)$ defined in
\eqref{eq:I.QuadrForm}. 
 
\subsection{The general construction  for subsets $\calS \subset \MM$} 

Throughout, our subsets $\calS$ will be path-connected, i.e.\ between any to
points $\mu_0,\mu_1\in \calS$ there exists a continuous path $\gamma \in
\rmC^0([0,1];\MM)$ with $\gamma(s)\in \calS$ for all $s\in [0,1]$. The
intrinsic length of $\gamma$ is defined by  
\begin{equation}
   \label{eq:IntrLength}
   L_\He(\gamma) = \sup \Bigset{\sum_{i=1}^N  \He\big(\gamma(s_i), 
    \gamma(s_{i-1})\big) }{N\in \N,\ 0=s_0{<}s_1{<}{\cdots} {<} s_{N-1}{<}s_N=1  },
\end{equation}
and $L_\He(\gamma)<\infty$ means that $\gamma$ is rectifiable in
$(\MM,\He)$. In that case, we can change the parametrization with a monotone
function $t:[0,1]\to [0,1]$ such that $\wt\gamma=\gamma \circ t$ has constant
speed, namely 
\begin{equation}
  \label{eq:ConstSpeed}
  \He\big(\wt\gamma(r),\wt\gamma(s)\big) =
  \He\big(\gamma(t(r)),\gamma(t(s))\big) = | r{-}s|\,L_\He(\gamma) \quad
\text{for all } r,s\in [0,1]. 
\end{equation}
To see this, consider the function $\ell:[0,1]\to [0,1]$ with
$\ell(t) = L_\He\big(\gamma|_{[0,t]}\big)/L_\He(\gamma)$, which is continuous,
non-decreasing and surjective. Now, we can choose any $t :[0,1]  \to [0,1]$ such
that $t\circ \ell=\mafo{id}_{[0,1]}$, i.e.\ $t(s) \in \ell^{-1}(\{s\})$. Using
the metric derivative  or speed) $|\dot{\wt\gamma}|_\He(s)$ as defined in
Section \ref{su:AbsCurves}, one then has the relation 
\begin{equation}
  \label{eq:Length.Quadr}
  L_\He(\gamma)^2=L_\He(\wt\gamma)^2 = \lim_{N\to \infty} \sum_{i=1}^N N \,\He\big(
\wt\gamma(i/N),\wt\gamma((i{-}1)/N)\big)^2 .
\end{equation}

Using the length $L_\He$, the Fisher-Rao distance $\FR_\calS$ for the subset
$\calS$ is defined by 
\begin{subequations}
  \label{eq:FR.2abs}
  \begin{align}
\FR_\calS(\mu_0,\mu_1) &:= \inf\Bigset{L_\He(\gamma) }{ \gamma\in
  \rmC^0([0,1];\calS),\ \gamma(0)=\mu_0,\ \gamma(1)=\mu_1) }
\\
  \label{eq:FR.2abs.b}
&= \inf\Bigset{ \int_0^1 \!\!\big(|\dot \gamma|_\He(s)\big)^2\dd s }{ \gamma
  \text{ $2$-absol.\  contin., } \gamma(0)=\mu_0,\ \gamma(1)=\mu_1) }.
  \end{align}
\end{subequations}
From the definition we immediately obtain the lower estimate 
\begin{equation}
\label{eq:He.leq.FR}
\FR_\calS(\mu_0,\mu_1)\geq \He(\mu_0,\mu_1)  \quad \text{for all } \mu_0,\mu_1
\in \calS,
\end{equation}
and equality can only hold if the geodesic curve $\gamma^\He_{\mu_0\to\mu_1}$
(cf.\ \eqref{eq:I.HellGeod}) is contained in $\calS$.

If $\calS$ is a smooth manifold that is given by a parameter $p\in D\subset X$,
where $X$ is a Banach space (e.g.\ $\R^m$) in the form  $ \calS = \bigset{\wh
  \mu(p) }{  p \in D\subset X }$, then the induced metric tensor, also called
Fisher's information matrix. can be reconstructed via 
\begin{equation}
\label{eq:FishM.from.He}
 \big\langle \wh \bbG(p)v, v\rangle = \lim_{\eps\to 0^+} \frac1{\eps^2}
 \He\big( \wh \mu(p), \wh\mu(p{+}\eps v)\big)^2 
\end{equation}

\subsection{Bhattacharya distance alias spherical Hellinger distance}
\label{su:BhattDist}

The simplest and still very important submanifold in $\MM$ is the set
of probability measures $\PP \subset \MM$. 

Since the Hellinger distance satisfies the general scaling property
\begin{equation}
  \label{eq:Scalingr0r1}
  \He(r_0^2\mu_0, r_1^2\mu_1)^2 = r_0r_1 \He(\mu_0,\mu_1)^2 +
\sigma^2 (r_0^2{-}r_0r_1)\mu_0(\Omega) + \sigma^2 (r_1^2{-}r_0r_1)\mu_1(\Omega),
\end{equation}
we can interpret the set $\MM$ of all (non-negative) measures as a
metric cone over the base space $\PP$, in the sense of
\cite[\S3.6]{BuBuIv01CMG}. This general geometric construction implies that
the induced Fisher-Rao distance $\FR_\calP$ on $\PP$, which is also
called \emph{Bhattacharya distance} $\Bh$ (cf.\ \cite{Rao45IAAE}), as well as
the exact form of the geodesics can be given explicitly, see
\cite[Sec.\,2]{LasMie19GPCA} for the details.  In the latter work this distance
is called the \emph{spherical Hellinger distance} because the base space
$\calP(\Omega)$ is called the spherical space of the cone. We obtain
\[
\Bh(\nu_0,\nu_1)= 2\sigma \arcsin\big( \frac1{2\sigma} \He(\nu_0,\nu_1)\big)
= \sigma \arccos\big(1 {-}\tdfrac1{2\sigma^2} \He(\nu_0,\nu_1)^2 \big) ,
\]
Note that $\He$ takes values in $[0,\sigma\sqrt{2}]$ (because
$\nu_j(\Omega)=1$), whereas $\Bh$ takes 
values in $[0, \sigma\,\pi/2]$. The maximum values are achieved if $\nu_0$ and
$\nu_1$ are mutual singular such that $\sqrt{\nu_0\nu_1}=0$, e.g.\ for Dirac
measures  $\nu_j=\delta_{\omega_j}$ with $\omega_0\neq \omega_1$. 

According to \cite[Thm.\,2.7]{LasMie19GPCA}, the geodesics take the form
\begin{align*}
\gamma^\Bh_{\nu_0\to\nu_1}(s) &= n(s) \gamma^\He_{\nu_0\to\nu_1}(t(s)) 
\quad \text{with } t(s) 
= \frac{\sin( s \delta)}{\sin\!\big((1{-}s)\delta\big) + \sin(s \delta)} \in [0,1],
\\
&\quad \delta= \tdfrac1\sigma \,\Bh(\nu_0,\nu_1)\in [0,\tfrac\pi2], \text{ and } n(s)= \big(
\frac{\sin\!\big((1{-}s)\delta\big) + \sin(s \delta)}{\sin(\delta)}  \big)^2
\in [1,2].  
\end{align*}
Recalling $\gamma^\He_{\nu_0\to\nu_1}(t)(\Omega)= 1 -
(t{-}t^2)\frac1{\sigma^2}\He(\nu_0,\nu_1)^2 = 1 -2 (t{-}t^2)(1{-}\cos\delta)$
we indeed find $ n(s)\gamma^\He_{\nu_0\to\nu_1}(t(s))(\Omega) \equiv 1$, i.e.\
$\gamma^\Bh_{\nu_0\to\nu_1}(s) \in \PP$. 

Of course, one can also go opposite and consider $\MM$ as a cone over $\PP$,
i.e.\ $\MM={[0,\infty[}\PP$. Then, we have the relation 
\begin{subequations}
  \label{eq:Hell.from.Bhatt}
 \begin{align}
  \label{eq:Hell.from.Bhatt.a}
 \He(r_0^2\nu_0,r_1^2 \nu_1)^2 &= r_0r_1 \He(\nu_0,\nu_1)^2 +
  \sigma^2(r_1{-}r_0)^2 
\\  \label{eq:Hell.from.Bhatt.b}
&= \sigma^2\big( (r_1{-}r_0)^2 + 2r_0r_1 \big( 
1{-}\cos\big(\tdfrac1\sigma \Bh(\nu_0,\nu_1)\big) \big) .
\end{align}
\end{subequations}

\begin{remark}[Unique characterization]
\label{re:UniqChar}
In \cite{KLMP13GDGC} the question is studied whether $\Bh$ is the only
Riemannian distance (up to a positive scalar factor) on a finite-dimensional,
smooth manifold that is invariant under all pushforwards (cf.\ Section
\ref{su:Pushforward}) with respect to smooth diffeomorphisms. It is shown that
this is true for dimension $n\geq 2$ but it may fail for $n=1$. It is unclear
whether the pathology for $n=1$ disappears if pushforwards for all measurable
homeomorphisms are considered.
\end{remark}

\subsection{General cones}
\label{su:GenCones}

A special property of the Hellinger distance is the scaling property
\eqref{eq:Scalingr0r1} that suggests that the transition between $\MM$ and
$\PP$ can be seen as the transition between 
the cone $\bbC_\calP \subset \MM$  and its base space $\calP\subset \PP$:  
\[
\calP \subset \PP \quad \text{and} \quad \bbC_\calP:= {[0,\infty[}\calP:= 
\bigset{r^2 \nu}{ r\in {[0,\infty[}, \ \nu \in \calP} =\calS,
\]
Here $\calP$ is chosen arbitrarily such that $(\calP,\FR_\calP)$ is a 
length space. The following result gives an explicit formula for $\FR_\calS$ in
terms of $\FR_\calP$. Whenever $\FR_\calP(\nu_0,\nu_1)\geq \sigma\pi$ we will
find $\FR_\calS(r_0^2\nu_0,r_1^2\nu_1)=\sigma(r_0+r_1)$ and the corresponding
geodesic curve is given by 
\[
\gamma^\calS_{r_0^2\nu_0 \to r_1^2\nu_1}(s) = \left\{\ba{cl}
  \big(r_0-(r_0{+}r_1)s\big)^2 \nu_0& \text{for }s\in [0,r_0/(r_0{+}r_1)], \\
 \big(r_0{+}r_1)s-r_0\big)^2 \nu_1 & \text{for }s\in [r_0/(r_0{+}r_1),1].
\ea  \right.
\]
For $\delta:=\tdfrac1\sigma\,\FR_\calP(\nu_0,\nu_1) < \pi$ geodesics in
$\calS$ can be expressed by geodesics in $\calP$ via 
\begin{align*}
\gamma^\calS_{r_0^2\nu_0\to r_1^2\nu_1}(s) &
= \wh r(s)^2\,\gamma^\calP_{\nu_0\to\nu_1}\big(\zeta(s)\big) \text{ with } 
 \zeta(s) = \frac1\delta \arcsin\Big( s\, 
\frac{r_1\sin \delta}{\wh r(s)} \Big) \\
&\quad \text{ and }\wh r(s)^2 =(1{-}s)^2 r_0^2+s^2 r_1^2 +2(s{-}s^2)
r_0r_1\cos_\pi \delta . 
\end{align*}
We refer to \cite[Sec.\,2.3]{LasMie19GPCA} for these formulas of the geodesics,
while the formula for $\FR_S$ given below is from
\cite[\S3.6]{BuBuIv01CMG}. Here we give a sketch of an alternative proof using
metric speeds.

\begin{theorem}[Fisher-Rao distance on cones]
\label{th:FRCones} If $(\calP,\FR_\calP)$ is a length space and $S=\bbC_\calP
\subset \MM$, then $(\calS,\FR_S)$ is a length space with 
\begin{equation}
  \label{eq:FRS.FRP}
  \FR_\calS(r_0^2\nu_0,r_1^2\nu_1)^2 = \sigma^2\Big(r_0^2+r_1^2 - 2r_0r_1
\cos_{\pi}\big( \tdfrac1\sigma\FR_\calP(\nu_0,\nu_1)\big) \Big),
\end{equation}
where $\cos_\pi(r)=\cos\big(\min\{|r|,\pi\}\big)$. 
\end{theorem}
\begin{proof}[Sketch of proof]
  A curve $s\mapsto \mu(s)=r(s)^2 \nu(s)\in \calS \subset \MM$ can having
  finite length has a metric speed a.e.\ in $[0,1]$. According to Theorem
  \ref{th:AbsContCurve} we can calculate the speed via Hellinger's quadratic
  form
\[
\big(|\dot \mu|_\He(s)\big)^2= \int_\Omega \xi_s^2 \dd \mu_s.
\]
Similarly, we can calculate the metric speed of
$s\mapsto \nu(s) \in \calP\subset \PP$. From \eqref{eq:Hell.from.Bhatt.a} 
we obtain 
\begin{align*}
|\dot\mu|_\He(s)^2 & = \lim_{h\searrow} \frac1{h^2} \,\He(\mu(s),\mu(s{+}h))^2 
\\
&= \lim_{h\searrow} \frac1{h^2} \Big( r(s{+}h)r(s) \He(\nu(s),\nu(s{+}h))^2 +
 \sigma^2\big( r(s{+}h)- r(s)\big)^2 \Big)
\\
&= r(s)^2 |\dot\nu|_\He(s)^2 +  \big(r'(s)\big)^2 .
\end{align*} 
Since $\FR_\calS(r_0^2\nu_0,r_1^2\nu_1)^2$ is given of the infimum over 
\[
\int_0^1|\dot\mu|_\He(s)^2\dd s = \int_0^1 \big(r(s)^2 |\dot\nu|_\He(s)^2 +
\sigma^2\big(r'(s)\big)^2 \big) \dd s 
\]
subject to the boundary conditions $\mu(j)=\mu_J=r_j^2 \nu_j$ for $j=0,1$, we
see obtain 
\[
\FR_\calS(r_0^2\nu_0,r_1^2\nu_1)^2= \inf\Bigset{\int_0^1\!\!\big( r^2
  y^2 +(\sigma r')^2\big) \dd s }{ r(0)=r_0,\ r(1)=r_1, \ \int_0^1\!\! y\dd s
  =\FR_\calP(\nu_0,\nu_1) }.
\]  
This minimization problem has been analyzed explicitly by
\cite[Thm.\,2]{LiMiSa16OTCR} for the case $\sigma=1$ (if one sets $\alpha=1$
and $\beta=4$ there). A crucial point is to realize that $ r y =\mafo{const}$
along minimizers. 
The case of general  $\sigma $ follows by scaling
replacing $y $ by $\sigma y$, thus rescaling $\FR_\calP$ by a factor $\sigma$, 
and pulling out the factor $\sigma^2$. 
 
This yields the desired formula \eqref{eq:FRS.FRP}. 
\end{proof}

\subsection{Product measures}
\label{su:ProdMeas}

In applications one is often interested in situations where the basic measure
space is a product space, viz.\ $\Omega= \Omega_1\ti \Omega_2$. Given  subsets
$\calS_1\subset \mathfrak M(\Omega_1)$ and $\calS_2\subset \mathfrak
M(\Omega_2)$ one is then interested in the Fisher-Rao distance for the subset 
\[
\calS_1\otimes \calS_2:=\bigset{\mu_1\oti \mu_2}{ \mu_1\in \calS_1, \mu_2\in
  \calS_2}.
\]   
The natural question is whether $\FR_{\calS_1\oti\calS_2}$ can be expressed of
estimated by $\FR_{\calS_1}$ and $\FR_{\calS_2}$. 

A positive and simple answer can be given in the case that $\calS_j$ are
contained in the probability measures $\mathfrak P(\Omega_j)$. 

\begin{proposition}[Product probability measures]
\label{pr:ProdProbMeas}
Assume that $\calP_j \subset \mathfrak P(\Omega_j)$ and that $\FR_{\calP_j}$ are
finite for $j=1,2$, then we have 
\begin{equation}
  \label{eq:FR.ProdProb}
  \FR_{\calP_1\oti\calP_2}(\nu_1\oti \nu_2,\eta_1\oti \eta_2)^2 = 
  \FR_{\calP_1}(\nu_1,\eta_1)^2 +  \FR_{\calP_2}(\nu_2, \eta_2)^2 \quad
  \text{for all } \nu_j,\eta_j\in \mathfrak P (\Omega_j). 
\end{equation}
\end{proposition}  
\begin{proof}
Since we are working with probability measures we can use the simple
representation \eqref{eq:HellProbProd} for the Hellinger distance of product
measures. We first observe that $\He(\nu_j,\eta_j)^2 \leq 2\sigma^2$ implies 
\begin{align*} 
\frac12\big(\He(\nu_1,\eta_1)^2+\He(\nu_2,\eta_2)^2\big)& \leq 
\He(\nu_1,\eta_1)^2+\He(\nu_2,\eta_2)^2- \frac1{2\sigma^2}
\He(\nu_1,\eta_1)^2\He(\nu_2,\eta_2)^2\\
& = 
\He(\nu_1\oti\eta_1,\nu_2,\eta_2)  \leq 
\He(\nu_1,\eta_1)^2+\He(\nu_2,\eta_2)^2.
\end{align*}
Hence the curve $s \mapsto \mu(s)$ is rectifiable if and only if the two curves
$s \mapsto \nu_j(s)$ are rectifiable. 

This allows us to calculate the metric speed for curves $\mu(s) =
\nu_1(s)\oti \nu_2(s)$ as follows. For a.a.\ $s\in [0,1]$ the three metric
derivatives $|\dot\mu|_\He$,  $ |\dot\nu_1|_\He$, and $|\dot\nu_2|_\He$ exists
and for those points we have . 
\begin{align*}
|\dot\mu|_\He(s)^2 &= \lim_{h\searrow 0} \frac1{h^2} \He\big(\mu(s),\mu(s{+}h)\big)^2 
\\
&\overset{\text{\eqref{eq:HellProbProd}}}= \lim_{h\searrow 0} \frac1{h^2}\Big(
\He\big(\nu_1(s),\nu_1(s{+}h)\big)^2 + \He\big(\nu_2(s),\nu_2(s{+}h)\big)^2\\
&\hspace*{6em} -
\frac1{2\sigma^2}
\He\big(\nu_1(s),\nu_1(s{+}h)\big)^2\He\big(\nu_2(s),\nu_2(s{+}h)\big)^2 \Big)
\\
&= \lim_{h\searrow 0} \frac1{h^2} \He\big(\nu_1(s),\nu_1(s{+}h)\big)^2
+ \lim_{h\searrow 0} \frac1{h^2} \He\big(\nu_2(s),\nu_2(s{+}h)\big)^2 - 0
\\
&= |\dot\nu_1|_\He(s)^2 +  |\dot\nu_2|_\He(s)^2.
\end{align*}

Thus, formula \eqref{eq:FR.ProdProb} follows immediately as calculating the
Fisher-Rao distances via the characterization in \eqref{eq:FR.2abs.b}. 
\end{proof}

This simple additive structure for the Fisher-Rao distance of product measures
disappears if one leaves the realm of probability measures. Relation
\eqref{eq:HellDistProd} can be rewritten in the form 
\begin{equation}
  \label{eq:HellDistProd2}
\begin{aligned}
  \He\big(\nu_1\oti \nu_2, \eta_1\oti \eta_2\big)^2 &
  = \tdfrac12\big(\ol\nu_2 {+} \ol\eta_2  \big)   \He(\nu_1,\eta_1)^2 
    + \tdfrac12\big(\ol\nu_1 {+} \ol\eta_1  \big)   \He(\nu_2,\eta_2)^2 
\\
&\quad 
 + \sigma^2 \big( \ol\nu_1{-} \ol\eta_1\big) \big( \ol\nu_2{-}
 \ol\eta_2\big) - \frac1{2\sigma^2} \He(\nu_1,\eta_1)^2  \He(\nu_2,\eta_2)^2 ,
\end{aligned}
\end{equation}
where $\ol\nu_j=\nu_j(\Omega_j)$ and $\ol\eta_j=\eta_j(\Omega_j)$. 
Hence, for curves $\mu(s)= \nu(s)\oti \eta(s)$ we obtain the
metric speed 
\[
|\dot\mu|_\He^2= m_\eta(s) |\dot \nu|_\He^2 + m_\nu(s) |\dot\eta|_\He^2 
 + \sigma^2  m'_\nu(s)m'_\eta(s)
\]
with $m_\mu(s) = \nu_s(\Omega)$ and $m_\eta(s) = \eta_s(\Omega)$.
Hence, there is a much stronger interaction between the measures in
$\calS_1\subset \mathfrak M(\Omega_1)$ and those in $\calS_2\subset \mathfrak
M(\Omega_2)$, and in the general case an explicit form for
$\FR_{\calS_1\oti\calS_2}$ seems out of reach. 

There is one case that can be treated, namely if $\calS_1$ and $\calS_2$ are
cones over $\calP_1$ and $\calP_2$, respectively. In this case, we have
\[
S_1\oti S_2= \bbC_{\calP_1} \oti \bbC_{\calP_2} = \bbC_{\calP_1\oti \calP_2} .
\]
Thus, we can first apply Proposition \ref{pr:ProdProbMeas} to obtain
$\FR_{\calP_1\oti \calP_2}$ and afterwards invoke Theorem \ref{th:FRCones}.

\section{Classical families of probability distributions}
\label{se:FamProbDistr}

As applications of the above theory we treat a few classical examples. 
For further applications we refer to \cite{Mayb16FRMC}.

\subsection{Translations of a measure} 
\label{su:TransMeasures}

As a first example we treat the case $\Omega =\R^n$, fix a $\mu\in \mathsf
(\R^n)$,  and using the 
diffeomorphisms $\Phi^y: x\mapsto x{+}y$ we define 
\[
\calS_\mafo{trans}(\mu):= \bigset{\Phi^y_\# \mu}{ y \in \R^n}.
\]
Observe that $\calS_\mafo{trans}(\mu)$ is not path connected if $\mu$ has a
discrete part, because for Dirac measures we have $\Phi^y
\delta_z=\delta_{z-y}$ and $\He( \delta_{x_1}, \delta_{x_2})= \sqrt2$ for
$x_1\neq x_2$. If $\mu$ is absolutely continuous with respect to the
Lebesgue measure, $\calS_\mafo{trans}(\mu)$ is path connected but the intrinsic
length may still be infinite for nontrivial curves. For this 
consider $n=1$ and $\wt\mu = \bm1_{[0,1]} \dd x $ giving the distances 
\[
\He(\Phi^y_{\#}\wt\mu,\Phi^z_{\#}\wt\mu)^2= \sigma^2\int_\R\big(\bm1_{[0,1]}(x{-}y)-
\bm1_{[0,1]}(x{-}z) \big) \dd x = 2 \sigma^2 \min\big\{ |y{-}x|, 1\big\}.
\]
Hence, for every non-constant curve $\gamma$ in $\calS_\mafo{trans}(\wt\mu)$ we
have $L_\He(\gamma)=\infty$, which implies $ \FR_{\calS_\mafo{trans}(\wt\mu)}(
\Phi^y_{\#}\wt\mu,\Phi^z_{\#}\wt\mu)=\infty$ for $y\neq z$.  

However, considering $\wh\mu= f \dd x$ with $\sqrt{f} \in \rmH^1(\R^d)$, it is
straightforward to show that 
\[
 \big( \FR_{\calS_\mafo{trans}(\wh\mu)}(
\Phi^y_{\#}\wh\mu,\Phi^z_{\#}\wh\mu)\big)^2 = (z{-}y)\cdot \bbA (z{-}y),
\]
where the induced translation invariant Riemannian metric on the parameter
space $D=\R^n$ is given by 
\[
\bbA:=\int_{\R^d} \frac{\sigma^2}{f}\;\nabla f \otimes \nabla f \dd x= 
\int_{\R^d} 4\sigma^2\:\nabla \sqrt{f} \otimes \nabla \sqrt{f} \dd x.
\]

Clearly, the induced distance on $\R^n$ is translation invariant. If there are
further symmetries (reflections or rotations) of the measure $\mu=f\dd x$ they
are reflected in the induced matrix $\bbA$ by applying Lemma \ref{le:Pushforward}.  

\subsection{Submanifold of Poisson measures on $\N_0^d$}
\label{su:Poisson}

For $\Omega= \N_0^d$ the multivariate Poisson distribution $\pi_\alpha$ for
$\alpha \in {[0,\infty[}^d$ is given by (where $n=(n_i)_i \in \N_0^d$)
\[
 \pi_\alpha \big(\{n\}\big) = \frac{\ee^{-\ol\alpha} \,\alpha^n}{n!} \quad
 \text{ with } \ \ol\alpha = \sum_{i=1}^d \alpha_i, \ \ 
\alpha^n=\prod_{i=1}^d \alpha_i^{n_i}, \  \text{ and } n!=\prod_{i=1}^d
 n_i!\;.
\]

The Hellinger distance between $\pi_\alpha$ and $\pi_\beta$ can easily
calculated by observing that $\sqrt{\pi_\alpha\pi_\beta}$ is a multiple of
$\pi_{\frac12(\alpha+\beta)}$, giving
\[
\He(\pi_\alpha,\pi_\beta)^2 = 2\sigma^2\big(1 -  \ee^{b(\alpha,\beta)}\big) \  
\text{ with } b(\alpha,\beta) = \sum_{i=1}^d \!\sqrt{\alpha_i\beta_i} -
\tfrac{\alpha_i{+}\beta_i}2 = - \frac12 \sum_{i=1}^d\!\big(\sqrt{\alpha_i}-
\sqrt{\beta_i}\big)^2 .
\]  

The Fisher metric tensor on the $d$-dimensional manifold
\[
\calS_\mafo{Poiss}(\N_0^d) = \bigset{ \pi_\alpha }{ \alpha \in {[0,\infty[}^d }
\]
is now most easily constructed by applying \eqref{eq:FishM.from.He}, giving
\[
\big\langle \bbG_\mafo{Poiss}(\alpha) v, v\big\rangle = \frac{\sigma^2}{4} 
\sum_{i=1}^d \frac{v_i^2}{\alpha_i}. 
\]

With this, the associated Fisher-Rao distance for $\calS_\mafo{Poiss}(\N_0^d)$
can be calculated explicitly and gives the Hellinger distance on the positive
orthant ${[0,\infty[}^d = \calM(\{1,..,d\})$ of $\R^d$, viz.\
\[
\FR_\mafo{Poiss}(\pi_\alpha,\pi_\beta)^2 = \sigma^2\sum_{i=1}^d
\big(\sqrt{\alpha_i}-\sqrt{\beta_i}\big)^2 .
\]
The sum structure of this formula is to be expected because the Poisson
distributions are tensor products of scalar Poisson distributions, such that
Proposition \ref{pr:ProdProbMeas} applies.  

Surprisingly, as in the case of the Bhattacharya (a.k.a.\ spherical Hellinger
distance), the Fisher-Rao distance can be expressed in terms of the Hellinger
distance itself, namely 
\[
\FR_\mafo{Poiss}(\pi_\alpha,\pi_\beta)^2 = - 2\sigma^2 \log\big( 1 -
\frac1{2\sigma^2}\He(\pi_\alpha, \pi_\beta)^2\big) \geq \He(\pi_\alpha, \pi_\beta)^2.
\]
Note that $\FR_\mafo{Poiss}(\pi_\alpha,\pi_\beta)$ can become arbitrarily
large, while $ \He(\pi_\alpha, \pi_\beta) \in [0,\sqrt2\,\sigma]$. 

\subsection{Submanifold of exponential distributions}
\label{su:Exponential}

We now consider the case $\Omega = {[0,\infty[}^n$, the parameter space
$A={]0,\infty[}^n \subset \R^n$,  and the probability
densities 
\[
\eps_\alpha (x) = p(\alpha) \ee^{-\alpha\cdot x} \quad \text{with } p(\alpha) =
\prod_{i=1}^n \alpha_i.
\]  
The exponential submanifold is then given by 
\[
\calS_\mafo{exp} = \bigset{\eps_\alpha }{ \alpha \in A={]0,\infty[}^n \subset \R^n }
\subset \rmL^1(\Omega)\cap \PP.
\]
Again the Hellinger distances are easily calculated in terms of the arithmetic
mean $a(\alpha,\beta)=\frac12(\alpha{+}\beta) \in A$ and the geometric mean
$g(\alpha,\beta)  = \big( \sqrt{\alpha_i\beta_I}\big)_i\in A$, namely 
\[
\He(\eps_\alpha,\eps_\beta)^2 = 2\sigma^2\Big(1 -  \frac{p\big(g(\alpha,\beta) \big)}
{p\big( a(\alpha,\beta)\big)}\Big) = 2\sigma^2\Big(1 -  \prod_{i=1}^n
\frac{\sqrt{\alpha_i\beta_i}}{ \frac12(\alpha_i{+}\beta_i)} \Big).
\]

The Fisher information matrix is easily obtained by using Fisher's logarithmic
derivative, namely  
\begin{align*}
\big\langle \bbG_\mafo{exp}(\alpha) v,v\big\rangle &= \frac{\sigma^2}4\int_\Omega 
\big|v{\cdot}\nabla_\alpha \log\!\big(\eps_\alpha(x)\big)\big|^2 \eps_\alpha(x) \dd x 
\\
&= \frac{\sigma^2}4 \int_\Omega \big|\sum_{i=1}^n v_i 
      (\tdfrac1{\alpha_i}{-}x_i)\big|^2 \eps_\alpha(x) \dd x 
= \frac{\sigma^2}4 \sum_{k=1}^n \big( \tdfrac{v_k}{\alpha_k}\big)^2. 
\end{align*}
Hence, the Fisher-Rao distance takes the form 
\[
\FR_{\exp}(\eps_\alpha,\eps_\beta)^2 = \frac{\sigma^2}4 \sum_{k=1}^n
\big(\log\alpha_k-\log\beta_k\big)^2= \frac{\sigma^2}4 \sum_{k=1}^n
\big(\log(\alpha_k/\beta_k)\big)^2 . 
\]
Again the sum structure follows because the multivariate exponential
distribution is the tensor product of one-dimensional exponential distributions. 

Moreover, we can use the scaling invariance of the exponential distributions
under mappings $\Phi^D(x) = D x $ with $D=\mafo{diag}\big(
\delta_i)_{i=1,..,n}$. Using that $\Phi^D_\# \eps_\alpha  = \eps_{D\alpha}$ we
find the invariance $\FR_{\exp}(\eps_\alpha,\eps_\beta)= \FR_{\exp}(\eps_{D\alpha},
\eps_{D\beta})$. This implies that $\FR_{\exp}(\eps_\alpha,\eps_\beta)$ can only
depend on $\alpha_k/\beta_k$.

\subsection{Gaussian distributions or multivariate normal distributions}
\label{su:Gaussian}

The Fisher-Rao distance between Gaussian distributions was one of the
motivation to introduce the the concept of Fisher matrix and Fisher-Rao
distance and the one-dimensional case (univariate normal distributions) was
already discussed in \cite{Rao45IAAE}. For the general multivariate case, no
explicit formula is known so far, and estimating the Fisher-Rao distance between
Gaussians from above or below is an active field of research in the area of
information geometry. We refer to \cite{StPoCo15BFRD, PiCoSt19FRIM,
  NieBar19GSI, PiStCo20FRDM, Niel23SAMF} for some recent works in this direction.

\begin{remark}
The situation is different if the Hellinger distance is 
replaced by the Wasserstein distance, because in the Otto-Wasserstein geometry
the geodesic curves between Gaussians remains in the class of Gaussians. This
leads to the so-called Bures-Wasserstein distance, see e.g.\
\cite{BhJaLi19BWDP, KrSpSu21SIBW, LCBBR22?VIWG}.
\end{remark}

We consider the case $\Omega=\R^d$ and use the standard representation of
Gaussian measures $G_{m,\Sigma}\in \mathfrak P(\R^d)$ with density 
\[
p_{m,\Sigma}(x) = \frac{1}{\sqrt{\det(2\pi\Sigma)}} 
\exp\Big(- \frac12 (x{-}m)\cdot \Sigma^{-1}(x{-}m) \Big) ,
\]
where $m\in \R^d$ is the mean and $\Sigma \in \R^{d\ti d}_{\mafo{spd}}$ is the
symmetric and positive definite covariance matrix. 

The Hellinger distance can be calculated easily as
$ \sqrt{ G_{m_0,\Sigma_0} G_{m_1,\Sigma_1} } $ is a multiple of a Gaussian with
covariance $2(\Sigma_0^{-1}{+}\Sigma_1^{-1})^{-1}$:
\[
\He( G_{m_0,\Sigma_0}, G_{m_1,\Sigma_1})^2
 = 2\sigma^2\Big( 1 - \frac{   \exp\big( -\frac14 (m_1{-}m_0) \cdot 
                   (\Sigma_0^{-1}{+}\Sigma_1^{-1}) (m_1{-}m_0)\big) }
  {\big(\det \Sigma_0\,\det\Sigma_1\big)^{1/4} 
   \big(\det(\frac12\Sigma_0^{-1}{+}\frac12\Sigma_1^{-1})\big)^{1/2}}\Big).
\] 
One can check that this formula is invariant under pushforwards (in the sense
of Section \ref{su:Pushforward}) under all affine transformations $\Phi(x) = A
x +x_*$ by noting that $\Phi_{\#}G_{m,\Sigma}= G_{\ol m, \ol \Sigma'}$ with
$\ol m=Am+x_*$ and $\ol \Sigma = A\Sigma A^*$. For instance, the distance only
depends on $m_1{-}m_0$. i.e.\ it is translation invariant. 

Setting $\calP_\mafo{Gauss}:=\bigset{G_{m,\Sigma}}{m\in \R^d,\ \Sigma \in
  \R^{d\ti d}_\mafo{spd}} $ and $\FR_\text{Gauss} = \FR_{\calP_\mafo{Gauss}} $,
the invariance under pushforwards with respect to affine transformations is sill
true for $\FR_\text{Gauss}$, see also \cite{PiStCo20FRDM} or ``Property 1'' in
\cite[P.\,4]{Niel23SAMF}. This is best seen by looking at the equation for
the geodesic curves inside $\calP_\mafo{Gauss}$. For this one calculates the
Fisher information matrix; using the parameters $(m,\Sigma)$ it takes the form 
\begin{equation}
  \label{eq:FishMatrGauss}
  \binom{v}{V} \cdot \bbG_\mafo{Gauss}(m,\Sigma)\binom v V= \sigma^2
  \big( v\cdot \Sigma^{-1} v + \frac12 \trace(\Sigma^{-1}V\Sigma^{-1} V)\big). 
\end{equation}
From this the geodesic equations for $s \mapsto (m(s),\Sigma(s))$ can be
derived by treating the inverse quadratic form
\[
H(m,\Sigma,\xi,\Xi):=\binom{\xi}{\Xi}\cdot \bbK_\mafo{Gauss}(m,\Sigma) 
\binom\xi\Xi = \frac1{4\sigma^2} \Big( \xi\cdot \Sigma\xi + 
  2 \trace(\Sigma\Xi\Sigma\xi) \Big)
\] 
as a Hamiltonian $H$, see \cite[Sec.\,4.1]{LMTZ25?EGHK} for the details. Here
$\xi(s)\in \R^d$ and $\Xi(s) \in \R^{d\ti d}_\mafo{sym}$ are the dual variables
corresponding to $\xi_s=\xi(s,\cdot)\in \rmL^2(\Omega,\gamma_s)$ in Theorem
\ref{th:AbsContCurve}.  After scaling the parameter along the geodesic curves
by the prefactor $ \sigma^2/2$, one arrives at
\begin{subequations}
  \label{eq:GeodEqGauss}
\begin{align}
  \label{eq:GeodEqGauss.a}
&m'= \tfrac{\sigma^2}2\rmD_\xi H(m,\Sigma,\xi,\Xi)= \Sigma \xi, &
& \Sigma'= \tfrac{\sigma^2}2\rmD_\Xi H(m,\Sigma,\xi,\Xi)= 2 \Sigma \Xi \Sigma,
\\
  \label{eq:GeodEqGauss.b}
&\xi'=- \tfrac{\sigma^2}2\rmD_m H(m,\Sigma,\xi,\Xi)=0,&
& \Xi'=-\tfrac{\sigma^2}2\rmD_\Sigma H(m,\Sigma,\xi,\Xi)= -2\Xi\Sigma \Xi
-\frac12 \xi\oti\xi.
\end{align}
\end{subequations}
The translation invariance is seen in the fact that $H$ does not depend on $m$,
which implies that $\xi$ is a constant along solutions by Noether's theorem.
Similarly, for all $ A \in \mafo{GL}(\R^d) $ the mapping 
\[
(m,\Sigma,\xi,\Xi) \mapsto \big( A m, A\Sigma A^*, A^{-*} \xi, A^{-*}\Xi
A^{-1}\big)
\]
leaves the Hamiltonian invariant, and Noether's theory leads to the conserved
quantities 
\[
J= \Sigma \Xi + \frac12 \,m\oti \xi \in \R^{d\ti d} \text{ \ (no symmetry)}.
\]
The $(1{+}d{+}d^2)$ conserved scalar quantities defined via  $H \in \R^1$,
$\mu\in \R^d$ and $J\in \R^{d\ti d}$ are enough to show that the geodesic
curves can be found, see \cite{CalOll91ESIG} or
\cite[Eqn.\,(15)]{PiStCo20FRDM}. However, knowing the geodesics means solving
an initial-value problem, while knowing the distance $\FR_\mafo{Gauss}\big( 
(m_0,\Sigma_0) , (m_1,\Sigma_1) \big)$ means to solve a boundary-value problem.

The case $d=1$, which was already treated in \cite{Rao45IAAE}, is by now
classical and can be related to hyperbolic geometry by introduction
$N=\sqrt{2\Sigma}$ and using the conservation laws 
\[
4\Sigma^2\Xi^2 +2\Sigma \xi^2=h_*=\mafo{const}, \quad\xi = \xi_*, \ \text{ and }
2\Sigma\xi + m\xi=j_*
\] 
leads to the condition that $(m,N)$ lies on the semi-circle 
\[
\big(m- j_*/\xi_*\big)^2 + N^2 = h_*/\xi_*^2. 
\]

Another easy case occurs for $\xi=\xi_*=0\in \R^d$, where now $d \in
\N_*$ is general. This implies $m(s)=m_0=m_1$ and corresponds to Gaussians with
the same center. From $\Sigma'=2 J\Sigma$ and $\Xi'=-2\Xi J$ and the boundary
conditions $\Sigma(i)=\Sigma_i$ for $i=0,1$, we obtain $J=-\frac12\log \big( 
\Sigma_1\Sigma_0^{-1}\big)= -\frac12\Sigma_0^{1/2}
\log\big(\Sigma_0^{-1/2}\Sigma_1\Sigma_0^{-1/2}\big) \Sigma_0^{-1/2}$ and find  
\[
\Sigma(s)  = \Sigma_0^{1/2} \big( \Sigma_0^{-1/2} \Sigma_1\Sigma_0^{-1/2}
\big)^s \Sigma_0^{1/2} . 
\]

A third case can be handled by using the cross-product theory of Section
\ref{su:ProdMeas} and joining the two cases from above, 
namely the case where $m_1{-}m_0$ is an eigenvector of
$\Sigma_1$ and $\Sigma_0$. Using the invariance under affine transformations
(rotations suffice) we can assume $m_1{-}m_0 = \delta e_1$. In that case, we
can find a solution of the Hamiltonian system in the form 
\begin{align}
\label{eq:m1m0EV}
&m(s) = m_0 + \alpha(s) e_1, && \Sigma(s)=\bma{cc} \gamma(s)&0\\
0&\Gamma(s)\ema \in \R^{1\ti 1}\oti \R^{(d-1)\ti(d-1)}.
\\
&\mu(s) = \beta e_1, && \nonumber
\Xi(s)=\bma{cc} \zeta(s)&0\\ 0 & \Pi(s)\ema \in \R^{1\ti 1}\oti
\R^{(d-1)\ti(d-1)}.
\end{align}
For $(\alpha,\gamma)$ the one-dimensional theory applies, while for
$(\Gamma,\Pi)$ the theory with the same center $0$ can be used. 

In these three cases the following results for
$\FR_\mafo{Gauss}(G_{m_0,\Sigma_0}, G_{m_1,\Sigma_1})$ are obtained, see
\cite{PiStCo20FRDM}. To present the results in a unified way we use the
function $M:\R_\geq\ti \R_>\to \R_\geq$ with  
\[
M(\Delta,\Lambda) := \sqrt{2}\:\log\bigg( \frac18\Big(\sqrt{\Delta +
  2(\Lambda^{1/4}{+}\Lambda^{-1/4})^2} + \sqrt{\Delta +
  2(\Lambda^{1/4}{-}\Lambda^{-1/4})^2} \Big)^2 \bigg)
\]
which essentially arises from hyperbolic theory and satisfies $M(0,\Lambda)=
2^{-1/2}\left| \log \Lambda\right| $, $M(\Delta,1)=\sqrt{\Delta} +
O(\Delta)_{\Delta\to 0}$, and $M(\Delta,1)= \sqrt2\,\log \Delta+
O(1)_{\Delta\to \infty}$.

\begin{theorem}[Fisher-Rao distance within Gaussians]
\label{th:FR.Gauss} \mbox{ }
 
(A) In the one-dimensional case $d=1$ we have the formula
\begin{equation}
  \label{eq:FR.Gauss.d=1}
  d=1:\quad \FR_\mafo{Gauss}(G_{m_0,\Sigma_0}, G_{m_1,\Sigma_1}) = \frac\sigma2 \,
  M\Big(\tdfrac{(m_1{-}m_0)^2}{\sqrt{\Sigma_0\Sigma_1}}, 
  \frac{\Sigma_1}{\Sigma_0}\Big). 
\end{equation}

(B) In the case with the same center $m_1=m_0$ and $d\geq 1$ we have 
\begin{equation}
  \label{eq:FR.Gauss.SameC}
  d\geq 1:\quad \FR_\mafo{Gauss}(G_{m_0,\Sigma_0}, G_{m_0,\Sigma_1})^2
  =\frac{\sigma^2}4 \sum_{n=1}^d M(0,\Lambda_n)^2 =
  \frac{\sigma^2}2 \sum_{n=1}^d \big(\log \Lambda_n  \big)^2 ,
\end{equation}
where $\Lambda_n=
\lambda_n(\Sigma_0^{-1/2}\Sigma_1\Sigma_0^{-1/2}) $ is the $n$-th eigenvalue of
the symmetric matrix $\Sigma_0^{-1/2}\Sigma_1\Sigma_0^{-1/2}$. 

(C) Let $\Lambda_n>0$ be as in (B) and assume $ \Sigma_0^{-1} (m_1{-}m_0) =
\Lambda_0 \Sigma_1^{-1}(m_1{-}m_0)$, then,
\begin{equation}
  \label{eq:FR.Gauss.m1m0EV}
 \begin{aligned} \FR_\mafo{Gauss}(G_{m_0,\Sigma_0}, G_{m_1,\Sigma_1})^2  &=
  \frac{\sigma^2}4\,M\Big(
   |\Sigma_0^{-1/2}(m_1{-}m_0)|\,|\Sigma_1^{-1/2}(m_1{-}m_0)|,
    \Lambda_1\Big)^2\\
&\qquad  +  \frac{\sigma^2}4 \sum_{n=2}^d M(0,\Lambda_n)^2.  
\end{aligned}
\end{equation}
\end{theorem}  
We refer to \cite[Sec.\,2.1]{PiStCo20FRDM} and \cite[Sec.\,1.2]{Niel23SAMF} and
the references therein for the details of proofs and the corresponding the
calculations.   
The sum structure for $\FR^2$ on formulas \eqref{eq:FR.Gauss.SameC} and 
\eqref{eq:FR.Gauss.m1m0EV} are again
consequences of the cross-product theory in Proposition \ref{pr:ProdProbMeas} 
because the Gaussians can be simultaneously transformed affinely to
have the same eigenbasis. 

However, we warn the reader that the general case cannot be handled by cross
products in all the cases that the initial and final Gaussian have the same
product structure. The reason is that the geodesic curves may lead the space of
such product measures, see the discussion in \cite{StPoCo15BFRD, PiCoSt19FRIM,
  PiStCo20FRDM}. In the present context this can be seen by looking at the
geodesic equations \eqref{eq:GeodEqGauss} involving the rank-one matrix
$\xi\oti \xi \in \R^{d\ti d}_\mafo{sym}$. A cross-product structure would mean
that $\Sigma(s)$ and $\Xi(s)$ would have the same block structure for all
$s\in[0,1]$. Then, this would also hold for $\xi\oti \xi $, but together with
the rank-one condition this implies that only one block can be
non-trivial. This essentially explains the condition in case (C) of the above
theorem.

\paragraph*{Acknowledgments.} The author is very grateful to
Fran\c{c}ois-Xavier Vialard for stimulating and helpful discussion about the
historical background.  The research was partially supported by Deutsche
Forschungsgemeinschaft (DFG) through the Berlin Mathematics Research Center
MATH+ (EXC-2046/1, DFG project no.\ 390685689) subproject ``DistFell''.

\footnotesize

\addcontentsline{toc}{section}{References}

\bibliographystyle{my_alpha}
\bibliography{alex_pub,bib_alex}

\newcommand{\etalchar}[1]{$^{#1}$}
\def\cprime{$'$}
\providecommand{\bysame}{\leavevmode\hbox to3em{\hrulefill}\thinspace}
\providecommand{\MR}{}
\begin{thebibliography}{11}\itemsep0.1em

\bibitem[AGS05]{AmGiSa05GFMS}
L.~Ambrosio, N.~Gigli, and G.~Savar{\'e}, \emph{Gradient flows in metric spaces
  and in the space of probability measures}, Lectures in Mathematics ETH
  Z\"urich, Birkh\"auser Verlag, Basel, 2005.

\bibitem[AJ{\etalchar{*}}15]{AJLS15IGSS}
N.~Ay, J.~Jost, H.~V.~L\^e, and L.~Schwachh\"ofer: \emph{Information geometry
  and sufficient statistics}. Probab. Theor. Relat. Fields \textbf{162} (2015)
  327--364.

\bibitem[BBI01]{BuBuIv01CMG}
D.~Burago, Y.~Burago, and S.~Ivanov, \emph{A course in metric geometry},
  Graduate Studies in Mathematics, vol.~33, American Mathematical Society,
  Providence, RI, 2001.

\bibitem[BBM16]{BaBrMi16UFRM}
M.~Bauer, M.~Bruveris, and P.~W.~Michor: \emph{Uniqueness of the {Fisher-Rao}
  metric on the space of smooth densities}. Bull. Lond. Math. Soc.
  \textbf{48}:3 (2016) 499--506.

\bibitem[BeB00]{BenBre00CFMS}
J.-D.~Benamou and Y.~Brenier: \emph{A computational fluid mechanics solution to
  the {M}onge-{K}antorovich mass transfer problem}. Numer. Math. \textbf{84}:3
  (2000) 375--393.

\bibitem[Bha42]{Bhat42DD}
A.~Bhattacharya: \emph{On discrimination and divergence}. Proc.\ Indian Sci.
  Congress \textbf{Part III} (1942) 13.

\bibitem[BJL19]{BhJaLi19BWDP}
R.~Bhatia, T.~Jain, and Y.~Lim: \emph{On the {Bures-Wasserstein} distance
  between positive definite matrices}. Expositiones Mathematicae \textbf{37}:2
  (2019) 165--191.

\bibitem[CaO91]{CalOll91ESIG}
M.~Calvo and J.~M.~Oller: \emph{An explicit solution of information geodesic
  equations for the multivariate normal model}. Stat. Risk Model.
  \textbf{9}:1-2 (1991) 119--138.

\bibitem[CH{\etalchar{*}}24]{CCHH24?FRGF}
J.~A.~Carrillo, D.~Z.~Huang, J.~Huang, and D.~Wei: \emph{{Fisher-Rao} gradient
  flow: geodesic convexity and functional inequalities}. Preprint
  \textbf{arXiv:2407.15693} (2024) .

\bibitem[CP{\etalchar{*}}18a]{CPSV18IDOT}
L.~Chizat, G.~Peyr\'e, B.~Schmitzer, and F.-X.~Vialard: \emph{An interpolating
  distance between optimal transport and {Fisher--Rao} metrics}. Found. Comput.
  Math. \textbf{18}:1 (2018) 1--44, (arXiv 2015).

\bibitem[CP{\etalchar{*}}18b]{CPSV18UOTG}
\bysame: \emph{Unbalanced optimal transport: geometry and {Kantorovich}
  formulation}. J. Funct. Analysis \textbf{274}:11 (2018) 3090--3123.

\bibitem[Fis21]{Fish21MFTS}
R.~A.~Fisher: \emph{On the mathematical foundations of theoretical statistics}.
  Phil. Trans. Royal Soc. A \textbf{222} (1921) 309--368.

\bibitem[Hel07]{Hell07OQFU}
E.~Hellinger, \emph{{Die Orthogonalinvarianten quadratischer Formen von
  unendlich vielen Variablen}}, Ph.D. thesis, Universit\"at G\"ottingen, 1907.

\bibitem[Hel09]{Hell09NBTQ}
\bysame: \emph{{Neue Begr\"undung der Theorie quadratischer Formen von
  unendlichvielen Ver\"anderlichen (in German)}}. J. Reine Angew. Math.
  \textbf{136} (1909) 210--271.

\bibitem[Kak48]{Kaku48EIPM}
S.~Kakutani: \emph{On equivalence of infinite product measures}. Annals Math.
  \textbf{49}:1 (1948) 214--224.

\bibitem[KL{\etalchar{*}}13]{KLMP13GDGC}
B.~Khesin, J.~Lenells, G.~Misio\l{}ek, and S.~C.~Preston: \emph{Geometry of
  diffeomorphism groups, complete integrability and geometric statistics}.
  Geom. Funct. Anal. \textbf{23} (2013) 334--366.

\bibitem[Kol30]{Kolm30UI}
A.~Kolmogoroff: \emph{{Untersuchungen \"uber den Integralbegriff}}. Math.
  Annalen \textbf{103} (1930) 654--696.

\bibitem[KSS21]{KrSpSu21SIBW}
A.~Kroshnin, V.~Spokoiny, and A.~Suvorikova: \emph{Statistical inference for
  {Bures-Wasserstein} barycenters}. Ann. Appl. Probab. \textbf{31}:3 (2021)
  1264--1298.

\bibitem[LaM19]{LasMie19GPCA}
V.~Laschos and A.~Mielke: \emph{Geometric properties of cones with applications
  on the {H}ellinger--{K}antorovich space, and a new distance on the space of
  probability measures}. J. Funct. Analysis \textbf{276}:11 (2019) 3529--3576.

\bibitem[LC{\etalchar{*}}22]{LCBBR22?VIWG}
M.~Lambert, S.~Chewi, F.~Bach, S.~Bonnabel, and P.~Rigollet: \emph{Variational
  inference via {Wasserstein} gradient ﬂows}. arXiv \textbf{2205.15902}
  (2022) .

\bibitem[LeC70]{Leca70APAN}
L.~LeCam: \emph{On the assumptions used to prove asymptotic normality of
  maximum likelihood estimates}. Ann. Math. Stat. \textbf{41}:3 (1970)
  802--828.

\bibitem[LM{\etalchar{*}}25]{LMTZ25?EGHK}
M.~Liero, A.~Mielke, O.~Tse, and J.-J.~Zhu: \emph{Evolution of {Gaussians} in
  the {Hellinger-Kantorovich-Boltzmann} gradient flow}. Comm. Pure Appl. Anal.
  (2025) , To appear, arXiv:250420400, WIAS preprint 3198.

\bibitem[LMS16]{LiMiSa16OTCR}
M.~Liero, A.~Mielke, and G.~Savar\'e: \emph{Optimal transport in competition
  with reaction -- the {H}ellinger--{K}antorovich distance and geodesic
  curves}. SIAM J. Math. Analysis \textbf{48}:4 (2016) 2869--2911.

\bibitem[LMS18]{LiMiSa18OETP}
\bysame: \emph{Optimal entropy-transport problems and a new
  {H}ellinger--{K}antorovich distance between positive measures}. Invent. math.
  \textbf{211} (2018) 969--1117.

\bibitem[May16]{Mayb16FRMC}
S.~J.~Maybank: \emph{A {Fisher-Rao} metric for curves using the information in
  edges}. J. Math. Imaging Vis. \textbf{54} (2016) 287--300.

\bibitem[MiZ25]{MieZhu25?ECHK}
A.~Mielke and J.-J.~Zhu: \emph{{Hellinger-Kantorovich} gradient flows: Global
  exponential decay of entropy functionals}. Preprint arXiv:2501.17049 (2025) ,
  WIAS preprint 3167.

\bibitem[NiB19]{NieBar19GSI}
F.~Nielsen and F.~Barbaresco (eds.), \emph{Geometric science of information},
  Lecture Notes in Computer Science, Springer Nature, 2019.

\bibitem[Nie23]{Niel23SAMF}
F.~Nielsen: \emph{A simple approximation method for the {Fisher-Rao} distance
  between multivariate normal distributions}. Entropy \textbf{25}:654 (2023)
  1--42.

\bibitem[Ott01]{Otto01GDEE}
F.~Otto: \emph{The geometry of dissipative evolution equations: the porous
  medium equation}. Comm. Partial Diff. Eqns. \textbf{26} (2001) 101--174.

\bibitem[PCS19]{PiCoSt19FRIM}
J.~Pinele, S.~I.~R.~Costa, and J.~E.~Strapasson, \emph{On the {Fisher-Rao}
  information metric in the space of normal distributions}, Geometric Science
  of Information (F.~Nielsen and F.~Barbaresco, eds.), Lecture Notes in
  Computer Science, Springer Nature, 2019, pp.~676--684.

\bibitem[PSC20]{PiStCo20FRDM}
J.~Pinele, J.~E.~Strapasson, and S.~I.~R.~Costa: \emph{The {Fisher-Rao}
  distance between multivariate normal distributions: special cases, bounds and
  applications}. Entropy \textbf{22}:404 (2020) 1--24.

\bibitem[Rao45]{Rao45IAAE}
C.~R.~Rao: \emph{Information and the accuracy attainable in the estimation of
  statistical parameters}. Bull. Calcutta Math. Soc. \textbf{37} (1945) 81--91.

\bibitem[RaV63]{RaoVar63DGP}
C.~R.~Rao and V.~S.~Varadarajan: \emph{Discrimination of {Gaussian} processes}.
  Sankhya: Indian J. Stat. A \textbf{25}:3 (1963) 303--330.

\bibitem[SPC15]{StPoCo15BFRD}
J.~E.~Strapasson, J.~P.~S.~Porto, and S.~I.~R.~Costa, \emph{On bounds for the
  {Fisher-Rao} distance between multivariate normal distributions}, AIP Conf.
  Proc. 1641, AIP Publishung LLC, 2015, Bayesian Inference and Maximum Entropy
  Methods in Science and Engineering (MaxEnt 2014), pp.~313--320.

\end{thebibliography}

\end{document}